\documentclass[11pt, reqno]{amsart}
\usepackage{amssymb}
\usepackage{graphicx} 
\usepackage{color}
\usepackage[textsize=small]{todonotes}
\numberwithin{equation}{section}

\newtheorem{theorem}{Theorem}[section]
\newtheorem{lemma}{Lemma}[section]
\newtheorem{remark}{Remark}[section]

\newtheorem{cor}{Corollary}[section]
\newtheorem{proposition}{Proposition}[section]

\begin{document}

\baselineskip 18pt

\newcommand{\inn}[2]{\langle #1, #2 \rangle}
\newcommand{\E}{\mathbb{E}}
\newcommand{\Eof}[1]{\mathbb{E}\left[ #1 \right]}
\newcommand{\Pof}[1]{\mathbb{P}\left[ #1 \right]}
\renewcommand{\H}{\mathbb{H}}
\newcommand{\R}{\mathbb{R}}
\newcommand{\sigl}{\sigma_L}
\newcommand{\BS}{\rm BS}
\newcommand{\p}{\partial}
\renewcommand{\P}{\mathbb{P}}
\newcommand{\var}{{\rm var}}
\newcommand{\cov}{{\rm cov}}
\newcommand{\beaa}{\begin{eqnarray*}}
\newcommand{\eeaa}{\end{eqnarray*}}
\newcommand{\bea}{\begin{eqnarray}}
\newcommand{\eea}{\end{eqnarray}}
\newcommand{\ben}{\begin{enumerate}}
\newcommand{\een}{\end{enumerate}}
\newcommand{\bit}{\begin{itemize}}
\newcommand{\eit}{\end{itemize}}

\newcommand{\bX}{\boldsymbol{X}}
\newcommand{\bY}{\boldsymbol{Y}}
\newcommand{\bt}{\boldsymbol{t}}
\newcommand{\bv}{\boldsymbol{v}}
\newcommand{\bx}{\boldsymbol{x}}
\newcommand{\by}{\boldsymbol{y}}
\newcommand{\bE}{\mathbf{e}}
\newcommand{\bw}{\mathbf{w}}
\newcommand{\bW}{\boldsymbol{W}}
\newcommand{\bB}{\boldsymbol{B}}
\newcommand{\bZ}{\boldsymbol{Z}}
\newcommand{\bH}{\mathbf{H}}
\newcommand{\bF}{\mathbf{F}}
\newcommand{\bG}{\mathbf{G}}
\newcommand{\bs}{\mathbf{s}}
\newcommand{\bsihat}{\widehat{\bs_i}}

\newcommand{\cL}{\mathcal{L}}
\newcommand{\cI}{\mathcal{I}}
\newcommand{\cR}{\mathcal{R}}

\newcommand{\mt}{\mathbf{t}}
\newcommand{\mS}{\mathbb{S}}

\newcommand{\bsB}{\boldsymbol{B}}
\newcommand{\bsb}{\boldsymbol{b}}
\newcommand{\bsK}{\boldsymbol{K}}
\newcommand{\bsR}{\boldsymbol{R}}
\newcommand{\bsxi}{\boldsymbol{\xi}}
\newcommand{\bseta}{\boldsymbol{\eta}}
\newcommand{\bszeta}{\boldsymbol{\zeta}}

\newcommand{\argmax}{{\rm argmax}}
\newcommand{\argmin}{{\rm argmin}}

\newcommand{\EE}{\mathbb{E}}
\newcommand{\tM}{\widetilde{M}}
\newcommand{\tE}{\tilde{\mathbb{E}}}
\newcommand{\hE}{\hat{\mathbb{E}}}
\newcommand{\tEof}[1]{\tilde{\mathbb{E}}\left[ #1 \right]}
\newcommand{\hP}{\hat{\mathbb{P}}}
\newcommand{\tP}{\tilde{\mathbb{P}}}
\newcommand{\tW}{\tilde{W}}
\newcommand{\tB}{\tilde{B}}
\newcommand{\1}{\mathbf{1}}
\renewcommand{\O}{\mathcal{O}}
\newcommand{\dt}{\Delta t}
\newcommand{\tr}{{\rm tr}}
\newcommand{\He}{{\rm He}}

\newcommand{\Xv}{X^{(v)}}
\newcommand{\Xvs}{X^{(v^*)}}
\newcommand{\Jv}{J^{(v)}}

\newcommand{\cG}{\mathcal{G}}
\newcommand{\cF}{\mathcal{F}}
\newcommand{\cLv}{\mathcal{L}^{(v)}}

\def\theequation{\thesection.\arabic{equation}}
\def\thetheorem{\thesection.\arabic{theorem}}

\renewcommand{\theequation}{\arabic{section}.\arabic{equation}}

\def\cprime{$'$}
\def\blue#1{\textcolor{blue}{#1}}

\title[Probability density of lognormal fSABR]{Probability density of lognormal fractional SABR model}

\author{Jiro Akahori, Xiaoming Song and Tai-Ho Wang}

\address{Jiro Akahori \newline
Department of Mathematical Sciences \newline
Ritsumeikan University \newline
Noji-higashi 1-1-1, Kusatsu, Shiga, 525-8577, Japan
}
\email{akahori@se.ritsumei.ac.jp}

\address{Xiaoming Song \newline
Department of Mathematics \newline
Drexel University\newline
32nd and Market Streets, Philadelphia, PA 19096
}
\email{song@math.drexel.edu}

\address{Tai-Ho Wang \newline
Department of Mathematics \newline
Baruch College, The City University of New York \newline
1 Bernard Baruch Way, New York, NY10010 \newline
and \newline
Department of Mathematical Sciences \newline
Ritsumeikan University \newline
Noji-higashi 1-1-1, Kusatsu, Shiga, 525-8577, Japan
}
\email{tai-ho.wang@baruch.cuny.edu}

\date{}
\maketitle

\begin{abstract}
Instantaneous volatility of logarithmic return in the lognormal fractional SABR model is driven by the exponentiation of a correlated fractional Brownian motion. Due to the mixed nature of driving Brownian and fractional Brownian motions, probability density for such a model is less studied in the literature. We show in this paper a bridge representation for the joint density of the lognormal fractional SABR model in a Fourier space. Evaluating the bridge representation along a properly chosen deterministic path yields a small time asymptotic expansion to the leading order for the probability density of the fractional SABR model. A direct generalization of the representation to joint density at multiple times leads to a heuristic derivation of the large deviations principle for the joint density in small time. Approximation of implied volatility is readily obtained by applying the Laplace asymptotic formula to the call or put prices and comparing coefficients.
\end{abstract}

\noindent{\em Keywords:} Asymptotic expansion, Lognormal fractional SABR model, Mixed fractional Brownian motion, Malliavin calculus, Bridge representation.


%
%

\newcommand{\F}{\mathcal{F}}
\renewcommand{\th}{\tilde{h}}

%
%

\section{Introduction}


The celebrated Black and Black-Scholes-Merton models have been the benchmark for European options on currency exchange, interest rates, and equities since the inauguration of the trading on financial derivatives.  However, empirical evidences have shown that the main drawback of these models is the assumption of constant volatility; the key parameter required in the calculation of option premia under such models. The volatility parameters induced from market data are in fact nonconstant across markets; dubbed as {\it volatility smile}. The Stochastic $\alpha\beta\rho$ (SABR for short hereafter) model, suggested by Hagan, Lesniewski, and Woodward in \cite{hlw}, is one of the  models, such as local volatility models, stochastic volatility models, and exponential L\'evy type of models etc, that attempts to capture the volatility smile effect.
Furthermore, as opposed to local volatility models, in SABR model the volatility smile moves in the same direction as the underlying with time, see \cite{hklw}.

The SABR model is depicted by the following system of stochastic differential equations (SDEs):
\bea
&& dF_t = \alpha_t F_t^\beta dW_t, \quad F_0 = F, \label{eqn:sabr-f} \\
&& d\alpha_t = \nu \alpha_t dZ_t, \quad \alpha_0 = \alpha, \label{eqn:sabr-a}
\eea
with $\beta \in [0,1]$, where $F_t$ is the forward price and $\alpha_t$ is the instantaneous volatility. $W_t$ and $Z_t$ are correlated Brownian motions with constant correlation coefficient $\rho$.
The SABR model is at times referred to as the lognormal SABR model when $\beta = 1$.
The SABR formula is an asymptotic expansion for the implied volatilities of call options with various strikes in small time to expiry. For reader's convenience, we reproduce the SABR formula in the following. Let $\sigma_{BS}(K,\tau)$ be the implied volatility of a vanilla option struck at $K$ and time to expiry $\tau$. The SABR formula states
\begin{equation} \label{eqn:sabr-formula}
\sigma_{BS}(K,\tau) = \nu \, \frac{\log(F/K)}{D(\zeta)}
\left\{1 + O(\tau) \right\}
\end{equation}
as time to expiry $\tau$ approaches 0. The function $D$ and the parameter $\zeta$ involved in \eqref{eqn:sabr-formula} are defined respectively as
\[
D(\zeta) = \log\left( \frac{\sqrt{1 - 2\rho\zeta + \zeta^2} + \zeta - \rho}{1-\rho}\right)
\]
and
\[
\zeta = \left\{\begin{array}{ll}
\frac\nu\alpha \, \frac{F^{1 - \beta} - K^{1 - \beta}}{1 - \beta} & \mbox{ if } \beta \neq 1; \\
& \\
\frac\nu\alpha \log\left(\frac FK \right) & \mbox{ if } \beta = 1.
\end{array}\right.
\]
Generally, the SABR formula is given one order higher, up to order $\tau$. Here we present only up to zeroth order for our own purpose.

The geometry of SABR model is isometrically diffeomorphic to the two dimensional hyperbolic space or the Poincar\'e plane. This isometry leads to a derivation of the SABR formula \eqref{eqn:sabr-formula} based on an expression of the heat kernel, known as the McKean kernel, on Poincar\'e plane. In particular, the lowest order term in \eqref{eqn:sabr-formula} has a geometric interpretation. The function $D$ is the shortest geodesic distance from the spot value $(F_0, \alpha_0)$ to the vertical line $F = K$ in the upper half plane $\{(F,\alpha) \in \mathbb{R}^2 : \alpha \geq 0\}$. Hence, the lowest order term in \eqref{eqn:sabr-formula} is indeed the ratio between the absolute value of logmoneyness, i.e.,  $\log(K/F_0)$, and the shortest geodesic distance from $(F_0,\alpha_0)$ to the vertical line $F = K$ in the upper half plane. We refer readers interested in this topic to \cite{hlw} for more detailed discussions. As expression for heat kernel on hyperblic space is concerned, Ikeda and Matsumoto in \cite{ikeda-matsumoto} provided a probabilistic approach and obtained, among other interesting results, a representation for the transition density of hyperbolic Brownian motion, i.e., the heat kernel over  the Poincar\'e plane. See Theorem 2.1 in \cite{ikeda-matsumoto} for details.

The aforementioned nice isometry between  SABR model and Poincar\'e plane breaks down if the volatility process, i.e., the $\alpha_t$ process in \eqref{eqn:sabr-f}, is driven by a fractional Brownian motion such as the second equation in \eqref{eqn:fSABR-s} considered in the paper. Moreover, due to the lack of Markovianity of fractional Brownian motions and thus the nonexistence of the forward and backward Kolmogorov equations, the classical asymptotic expansion approaches, such as the heat kernel or WKB expansion, are no longer applicable either. In this regard, the probabilistic approach in \cite{ikeda-matsumoto} is more applicable and tractable when dealing with processes driven by fractional Brownian motions.

The volatility process is generally conceived behaving ``fractionally" in that the driving noise is a  fractional process, e.g., a fractional Brownian motion with Hurst exponent other than a half. For a far from exhausting list, models that attempt to incorporate the fractional feature of volatility include:
%
the ARFIMA model in \cite{granger} and the FIGARCH model \cite{figarch} for discrete time models; the long memory stochastic volatlity model in \cite{comte-renault} and the affine fractional stochastic volatilility model in \cite{comte-coutin-renault} for continuous time models.
Somewhat on the contrary, in a recent study in \cite{gatheral-jaisson-rosenbaum}, the Hurst exponent $H$ is estimated as being less than a half; thereby indicating antipersistency as opposed to persistency of the volatility process. It is also worth mentioning that generalizations of Heston model to fractional version have been recently considered in \cite{rosenbaum} and \cite{jack}.
Heston related models are usually dealt with via the characteristic and/or moment generating functions. However, in this paper we take the approach following closely the methodology in \cite{ikeda-matsumoto}.

In order to embed the empirically observed fractional feature of the volatility process into the classical SABR model, we suggest in this paper a fractional version of the SABR model as in \eqref{eqn:fSABR-s}. Modulo a mean-reversion component, this model aligns with the model statistically tested in \cite{gatheral-jaisson-rosenbaum}. The main observation in \cite{gatheral-jaisson-rosenbaum} is that, using square root of the realized/integrated variance as a proxy for the instantaneous volatility, the logarithm of the volatility process behaves like a fractional Brownian motion in almost any time scale of frequency. The Hurst exponent $H$ inferred from the time series data is less than a half; indeed, $H \approx 0.1$. This observation of small Hurst exponent in the volatility process makes the analysis of the model more technical and challenging from stochastic analysis point of view. To our knowledge, most of the small time asymptotic expansions for processes driven by fractional Brownian motions have restrictions on the Hurst exponent $H$ of the driving fractional Brownian motion, mostly $H \geq \frac14$. One of the advantage of the approach undertaken in the current paper is that it works without restriction on the Hurst exponent $H$. The key ingredient is a representation in a Fourier space, which we call the bridge representation in Section \ref{sec:bridge-rep}, for the joint density of log spot and volatility, see \eqref{eqn:bridge-rep}.

A small time asymptotic expansion of the joint density is readily obtained from the bridge representation. The idea is to approximate the conditional expectation in the bridge representation by a judiciously chosen deterministic path since, conditioned on the initial and terminal points, at each point in time a Gaussian process will not wander too far away from its expectation.
As long as an asymptotic expansion for the density of the underlying asset is available, to obtain an expansion for implied volatility is almost straightforward by basically comparing the coefficients with a similar expansion obtained by using the lognormal density on the Black or the Black-Scholes-Merton side.

The methodology of deriving the bridge representation \eqref{eqn:bridge-rep} can be generalized directly to obtain a bridge representation for the joint density of multiple times; hence inducing a representation for finite dimensional distributions of the fractional SABR model, see Theorem \ref{thm:rep-multi}. Based on this bridge representation for finite dimensional distributions, Section \ref{sec:ldp} is devoted to a heuristic yet appealing derivation of the large deviations principle for the joint density of the fractional SABR model in small time. This large deviations principle in a sense can be regarded as defining a ``geodesic distance" over the fractional SABR plane since, as we shall show in Section \ref{sec:ldp}, it recovers the energy functional on the Poincar\'e plane when $H = \frac12$. We leave the rigorous proof of the large deviations principle in a future work. An immediate consequence of this large deviation principle is the fractional SABR formula (to the lowest order) \eqref{eqn:fSABR} which recovers the classical SABR formula when $H = \frac12$. The fractional SABR formula \eqref{eqn:fSABR} pertains the guiding principle that the lowest order term in the implied volatility expansion is given by the ratio between the absolute value of the logmoneyness and the geodesic distance to the vertical line $F = K$.

The rest of the paper is organized as follows. The fractional SABR model is specified and the bridge representation for joint density is shown in Section \ref{sec:bridge-rep}. Sections \ref{sec:density-exp} and \ref{sec:implied-vol} provide small time asymptotic expansions of the joint density and of the implied volatilities respectively. Section \ref{sec:ldp} presents the bridge representation for finite dimensional distributions and the large deviations principle. Finally, the paper concludes in Section \ref{sec:conclusion} with discussions.

%
%

\section{Model specification} \label{sec:bridge-rep}
Throughout the text, $B=\{B_t, t\geq 0\}$ and $W=\{W_t, t\geq 0\}$ denote two independent standard Brownian motions defined on the filtered probability space $(\Omega,\cF_t,\P)$ satisfying the usual conditions. Let $B^H=\{B_t^H,t\geq 0\}$ be a fractional Brownian motion with Hurst exponent $H \in (0,1)$ generated by $B$ (see \cite{DU}), i.e.,
\[
B^H_t = \int_0^t K_H(t,s) dB_s,
\]
where $K_H$ is the Molchan-Golosov kernel
\begin{equation}
K_H(t,s) = c_H (t-s)^{H-\frac{1}{2}}F\left(H-\frac{1}{2},\frac{1}{2}-H,H+\frac{1}{2};1-\frac{t}{s}\right)\mathbf{1}_{[0,t]}(s), \label{eqn:molchon-golosov}
\end{equation}
with $c_H=\left[\frac{2H\Gamma\left(\frac{3}{2}-H\right)}{\Gamma(2-2H)\Gamma\left(H+\frac12\right)}\right]^{1/2}$ and $F$ is the Gauss hypergeometric function. Also, the autocovariance function of a fractional Brownian motion is denoted by $R(t,s)$ and defined as
\begin{equation}
R(t,s) =\mathbb{E}(B^H_tB^H_s)= \frac12 \left( t^{2H} + s^{2H} - |t-s|^{2H} \right).  \label{eqn:Rts}
\end{equation}
Lastly, we assume that all random variables and stochastic processes are defined on $(\Omega,\cF_t,\P)$.

\subsection{The model}
We study the following lognormal fractional SABR (fSABR) model in a risk neutral probability (for simplicity, interest and dividend rates are both assumed zero in this paper):
\begin{equation}\label{eqn:fSABR-s}
\begin{cases}
 S_t=s_0+\int_0^t\alpha_r S_r (\rho dB_r + \bar\rho dW_r), \\
 \\
 \alpha_t = \alpha_0 e^{\nu B^H_t}, 
 \end{cases}
\end{equation}
where $s_0$ and $\alpha_0$ are the given time zero (current observed) values for the processes $S$ and $\alpha$ respectively, $\rho \in (-1,1)$ and $\bar\rho = \sqrt{1 - \rho^2}$.

 In other words, the underlying price $S_t$ follows a stochastic volatility model with the (instantaneous) volatility process $\alpha_t$, and $\alpha_t$ is given by the exponentiation of a correlated fractional Brownian motion. 
The main purpose of this section is to derive the bridge representations \eqref{eqn:bridge-rep} and \eqref{eqn:bridge-rep-s} for the joint densities of $(S_t,\alpha_t)$. The bridge representation is the crucial starting line in obtaining expansions and approximations of the joint densities to be discussed in Section \ref{sec:density-exp}. 

By making the change of variables
\begin{eqnarray*}
&& X_t = \ln{ S_t}, \qquad Y_t = \alpha_t,
\end{eqnarray*}
the system \eqref{eqn:fSABR-s} can be written more explicitly as 
\begin{equation}\label{eqn:fSABR-x}
\begin{cases}
 X_t = x_0+y_0 \int_0^t e^{\nu B^H_s} (\rho dB_s + \bar\rho dW_s) - \frac{y_0^2}2 \int_0^t e^{2 \nu B^H_s } ds,  \\
 \\
Y_t = y_0 e^{\nu B^H_t},
\end{cases}
\end{equation}
where $x_0=\ln{s_0}$ and $y_0=\alpha_0$.

\subsection{Malliavin calculus with respect to Brownian motion}
We provide some preliminaries on Malliavin calculus with respect to the two Brownian motions $B$ and $W$ in this subsection. We refer the reader to \cite{hu16} and  \cite{N06} for more details.

For any fixed $T>0$, let $\mathbf{H}=L^2([0,T])$ be the separable Hilbert space of all
square integrable real-valued functions on the interval $[0,T]$
with scalar product denoted by $\langle
\cdot,\cdot\rangle_\mathbf{H}$. The norm of an element $h\in
\mathbf{H}$ will be denoted by $\Vert h\Vert_\mathbf{H}$. For any
$h\in \mathbf{H}$, we put $W(h)=\int_0^Th(t)dW_t$ and $B(h)=\int_0^Th(t)d{B}_t$.

For any $m,\, n\in\mathbb{N}$, denote by $C_p^\infty(\mathbb{R}^{m+n})$ the set of all infinitely
 differentiable functions $g: \mathbb{R}^{m+n}\rightarrow
\mathbb{R}$ such that $g$ and all of its partial derivatives have
polynomial growth. We    make use of the notation $\partial_i g=\frac{\partial
g}{\partial x_i}$ whenever $g\in C^1(\mathbb{R}^{m+n})$.

Let $\mathcal{S}$ denote the class of smooth and cylindrical random variables such
that a random variable $F\in \mathcal {S}$ has the form
\begin{equation}\label{smooth}
F=g(W(h_1),\dots,W(h_m),B(k_1),\dots,B(k_n)),
\end{equation}
where $g$ belongs to $C_p^\infty(\mathbb{R}^{m+n})$, $h_1,\dots,h_m$ and $k_1, \dots, k_n$
are in $\mathbf{H}$, and $m, n\in\mathbb{N}$.

For a smooth and cylindrical random variable $F$  of the
form (\ref{smooth}), its Malliavin derivative with respect to $W$  is the $\mathbf{H}$-valued random variable
given by
\[
D_t^1F=\sum_{i=1}^m\partial_i g(W(h_1),\dots,W(h_m),B(k_1),\dots, B(k_n))h_i(t), \, t\in[0,T],
\]
and respectively its Malliavin derivative with respect to $B$ is given by 
\[
D_t^2F=\sum_{i=1}^n\partial_{m+i} g(W(h_1),\dots,W(h_m),B(k_1),\dots, B(k_n))k_i(t), \, t\in[0,T].
\]
For any $p\geq1$, we will denote the domain of $D$ in $L^p(\Omega)$
by $\mathbb{D}^{1,p}$, meaning that $\mathbb{D}^{1,p}$ is the
closure of the class of smooth and cylindrical  random variables $\mathcal{S}$ with
respect to the norm
\[
\Vert F\Vert_{1,p}=\left(\EE|F|^p+\EE\left(\Vert
D^1F\Vert_\mathbf{H}^2+\Vert
D^2F\Vert_\mathbf{H}^2\right)^{\frac{p}2}\right)^{\frac{1}{p}}.
\]

We tailor Theorem 2.1.2 in \cite{N06} to the following lemma which yields a result on the absolute continuity of the law of a random vector with respect to the Lebesgue measure.
\begin{lemma}\label{lem:absolute} Let $F=(F_1,F_2)$ be a random vector in $\mathbb{D}^{1,2}$. If the Malliavin matrix $\gamma := (\langle D^1F_i, D^1F_j\rangle_\mathbf{H}+\langle D^2F_i, D^2F_j\rangle_\mathbf{H})_{1\leq i,j\leq 2}$ of $F$ is invertible a.s.. Then the law of $F$ is absolutely continuous with respect to the Lebesgue measure on $\mathbb{R}^2$. Consequently, the joint density of the random variables $(F_1,F_2)$ exists.
\end{lemma}

\subsection{Bridge representation for the joint density} In this subsection, we show the existence of the joint density of $(X_t,Y_t)$ for any $t>0$ by using Malliavin calculus. We also give a bridge representation for the joint density by adapting the methodology introduced in Ikeda and Matsumoto \cite{ikeda-matsumoto}. 

\begin{theorem} \label{thm:bridge-rep-p}
For any $t>0$, the law of $(X_t,Y_t)$ satisfying  \eqref{eqn:fSABR-x} is absolutely continuous with respect to the Lebesgue measure on $\mathbb{R}^2$. Moreover, the joint probability density $p(t;x,y)$ of $(X_t, Y_t)$ has the following bridge representation
\bea\label{eqn:bridge-rep}
   && p(t;x,y) \nonumber\\
  &=& \frac1{y\sqrt{2\pi \nu^2 t^{2H}}} e^{-\frac{\left(\ln{(y/y_0)}\right)^2}{2\nu^2 t^{2H}}} \times \nonumber\\
  && \frac1{2\pi} \int_{\mathbb{R}} \Eof{\left. e^{i\left(x - x_0 - \rho y_0 \int_0^t e^{\nu B^H_s} dB_s + \frac{y_0^2v_t}2  \right)\xi} e^{-\frac{\bar\rho^2 y_0^2 v_t\xi^2}2 }\right|B_t^H=\frac{\ln{(y/y_0)}}{\nu}} d\xi.
\eea
where $v_t=\int_0^t e^{2\nu B^H_s} ds$ and  $i = \sqrt{-1}$.
\end{theorem}
\begin{remark}
The bridge representation \eqref{eqn:bridge-rep} can be regarded as a generalization of the well-known McKean kernel, namely, the classical heat kernel over a 2-dimensional hyperbolic space. For reader's reference, the McKean kernel $p_{\H^2}(t; x,y)$ reads
\[
 p_{\H^2}(t; x, y) = \frac{\sqrt2 e^{-t/8}}{(2\pi t)^{3/2}} \int_{d}^\infty \frac{\xi e^{-\xi^2/2t}}{\sqrt{\cosh\xi - \cosh d}} d\xi,
\]
where $d = d(x, y; x_0,y_0)$ is the geodesic distance from $(x, y)$ to $(x_0,y_0)$. The geodesic distance satisfies $\cosh d(x, y; x_0,y_0) = \frac{(x - x_0)^2 + y^2 + y_0^2}{2y y_0}$. Note that the McKean kernel is a density with respect to the Riemannian volume form $\frac1{y^2} dx dy$. Indeed, in the case where $H = \frac12$, $\nu =1$ and $\rho = 0$, Ikeda-Matsumoto in \cite{ikeda-matsumoto} showed how to recover the McKean kernel from \eqref{eqn:bridge-rep}. See also Cheng and Wang \cite{cheng-wang} for a different representation in terms of Bessel bridge for the hyperbolic heat kernel. 
\end{remark}
\begin{proof} Notice that we can rewrite  \eqref{eqn:fSABR-x} as 
\begin{equation*}
\begin{cases}
 X_t = x_0+ \int_0^t Y_s (\rho dB_s + \bar\rho dW_s) - \frac{1}2 \int_0^t Y_s^2 ds,  \\
 \\
Y_t = y_0 e^{\nu B^H_t}.
\end{cases}
\end{equation*}
Now we fix any $T\geq t$. Then according to Sections 2.2 and 5.2 in \cite{N06}, the Malliavin derivatives of $X_t$ and $Y_t$ are given as follows
\begin{eqnarray*}
D^1_\theta Y_t&=& 0,\\
D^2_\theta Y_t&=& y_0 \nu e^{\nu B_t^H}K_H(t,\theta) \mathbf{1}_{[0,t]}(\theta)
\end{eqnarray*}
and
\begin{eqnarray*}
D^1_\theta X_t&=&\bar\rho Y_\theta\mathbf{1}_{[0,t]}(\theta)=\bar\rho y_0e^{\nu B_\theta^H}\mathbf{1}_{[0,t]}(\theta), \\
D^2_\theta X_t&=&\left(\rho Y_\theta+\int_\theta^t\rho D^2_\theta Y_sdB_s+\int_\theta^t\bar\rho D^2_\theta Y_sdW_s-\int_\theta^tY_sD_\theta^2Y_s ds\right)\mathbf{1}_{[0,t]}(\theta)\\
&=&\left(\rho y_0e^{\nu B_\theta^H}+\rho y_0\nu \int_\theta^te^{\nu B_s^H}K_H(s,\theta)dB_s+\bar\rho y_0\nu \int_\theta^te^{\nu B_s^H}K_H(s,\theta)dW_s\right)\mathbf{1}_{[0,t]}(\theta)\\
&&-y_0^2\nu\int_\theta^te^{2\nu B_s^H}K_H(s,\theta)ds\, \mathbf{1}_{[0,t]}(\theta).
\end{eqnarray*}
Thus, the Malliavin matrix $\gamma$ of $(X_t,Y_t)$ is given by 
\[
\gamma=\left(\begin{array}{cc}
\gamma_{11}&\gamma_{12}\\
\gamma_{21}&\gamma_{22}
\end{array}\right),
\]
where
\begin{eqnarray*}
\gamma_{11}&=&\int_0^t(D^1_\theta X_t)^2d\theta+\int_0^t(D^2_\theta X_t)^2d\theta\\
&=&\int_0^t\bar\rho^2 y_0^2e^{2\nu B_\theta^H}d\theta+\int_0^t\left(\rho y_0e^{\nu B_\theta^H}+\rho y_0\nu \int_\theta^te^{\nu B_s^H}K_H(s,\theta)dB_s\right.\\
&&\qquad\qquad\qquad\left.+\bar\rho y_0\nu \int_\theta^te^{\nu B_s^H}K_H(s,\theta)dW_s-y_0^2\nu\int_\theta^te^{2\nu B_s^H}K_H(s,\theta)ds\right)^2d\theta,
\end{eqnarray*}
\begin{eqnarray*}
\gamma_{12}=\gamma_{21}&=&\int_0^tD^1_\theta X_tD^1_\theta Y_td\theta+\int_0^tD^2_\theta X_t D^2_\theta Y_td\theta\\
&=&y_0\nu e^{\nu B_t^H}\int_0^tK_H(t,\theta)\left(\rho y_0e^{\nu B_\theta^H}+\rho y_0\nu \int_\theta^te^{\nu B_s^H}K_H(s,\theta)dB_s\right.\\
&&\quad\left.+\bar\rho y_0\nu \int_\theta^te^{\nu B_s^H}K_H(s,\theta)dW_s-y_0^2\nu\int_\theta^te^{2\nu B_s^H}K_H(s,\theta)ds\right)d\theta,
\end{eqnarray*}
and 
\begin{eqnarray*}
\gamma_{22}&=&\int_0^t(D^1_\theta Y_t)^2d\theta+\int_0^t(D^2_\theta Y_t)^2d\theta\\
&=&y_0^2\nu^2e^{2\nu B_t^H}\int_0^tK_H(t,\theta)^2d\theta.
\end{eqnarray*}
Then it follows from the Cauchy-Schwarz inequality that almost surely
\begin{eqnarray*}
\gamma_{12}^2&<&y_0^2\nu^2e^{2\nu B_t^H}\int_0^tK_H(t,\theta)^2d\theta\times\\
&&\int_0^t\left(\rho y_0e^{\nu B_\theta^H}+\rho y_0\nu \int_\theta^te^{\nu B_s^H}K_H(s,\theta)dB_s\right.\\
&&\qquad\qquad\qquad\left.+\bar\rho y_0\nu \int_\theta^te^{\nu B_s^H}K_H(s,\theta)dW_s-y_0^2\nu\int_\theta^te^{2\nu B_s^H}K_H(s,\theta)ds\right)^2d\theta\\
&\leq&\gamma_{22}\cdot\gamma_{11},
\end{eqnarray*}
which implies that the Malliavin matrix $\gamma$ is invertible a.s.. Hence, by Lemma \ref{lem:absolute} the law of $(X_t,Y_t)$ is absolutely continuous with respect to the Lebesgue measure on $\mathbb{R}^2$.

Next, we calculate the joint probability density $p(t;x,y)$ of $(X_t,Y_t)$ as follows. For any bounded and continuous function $f$ defined on $\mathbb{R}^2$, we have
\bea\label{eqn:density-init}
  && \Eof{f(X_t,Y_t)} \nonumber\\
  &=& \Eof{f\left(x_0 + y_0 \int_0^t e^{\nu B^H_s} (\rho dB_s + \bar\rho dW_s) - \frac{y_0^2v_t}2  , y_0 e^{\nu B^H_t}  \right)}.
\eea
Note that conditioned on $\F^B_t$, $y_0 \bar\rho \int_0^t e^{\nu B^H_s} dW_s$ is normally distributed since $W$ and $B^H$ are independent. Moreover,
\beaa
  && \Eof{\left. y_0 \bar\rho \int_0^t e^{\nu B^H_s} dW_s \right| \F^B_t} = 0, \\
  && \Eof{\left. \left( y_0 \bar\rho \int_0^t e^{\nu B^H_s} dW_s \right)^2 \right|\F^B_t} = y_0^2 \bar\rho^2 \int_0^t e^{2\nu B^H_s} ds = y_0^2 \bar\rho^2 v_t.
\eeaa
From \eqref{eqn:density-init}, it follows  by conditioning on $\F^B_t$ that 
\bea\label{eqn:density-cont}
&& \Eof{f(X_t,Y_t)} \nonumber\\
  &=& \Eof{\E\left[\left.f\left(x_0 + y_0 \bar\rho \int_0^t e^{\nu B^H_s} dW_s + y_0 \rho \int_0^t e^{\nu B^H_s} dB_s - \frac{y_0^2 v_t}2, y_0 e^{\nu B^H_t} \right) \right|\F^B_t\right]} \nonumber\\
  &=& \Eof{\int\left\{f\left(x_0 + \xi + y_0 \rho \int_0^t e^{\nu B^H_s} dB_s - \frac{y_0^2v_t}2 , y_0 e^{\nu B^H_t} \right) \frac{e^{-\frac{\xi^2}{2 y_0^2\bar\rho^2 v_t}}}{\sqrt{2\pi y_0^2\bar\rho^2 v_t}} \right\}d\xi}\nonumber \\
 &=& \Eof{\int\left\{\frac1{\sqrt{2\pi y_0^2\bar\rho^2 v_t}}f\left(x, y_0 e^{\nu B^H_t} \right) e^{-\frac{\left(x - x_0 - y_0 \rho \int_0^t e^{\nu B^H_s} dB_s + \frac{y_0^2v_t}2  \right)^2}{2 y_0^2\bar\rho^2 v_t}} \right\}dx} \nonumber\\  
 &=&\int_{\mathbb{R}^2}f(x,y)\Eof{\frac1{\sqrt{2\pi y_0^2\bar\rho^2 v_t}}\left.e^{-\frac{\left(x- x_0 - y_0 \rho \int_0^t e^{\nu B^H_s} dB_s + \frac{y_0^2 v_t}2 \right)^2}{2 y_0^2\bar\rho^2 v_t}}\right|B_t^H=\frac{\ln{(y/y_0)}}{\nu}}\nonumber\\
 &&\qquad\qquad\qquad \times \frac{1}{y\sqrt{2\pi \nu^2t^{2H}}}e^{-\frac{\left(\ln{y}-\ln{y_0}\right)^2}{2\nu^2t^{2H}}}dx\, dy.
\eea
By using the identity 
\beaa
  && e^{-\frac{u^2}{2y_0^2 \bar\rho^2 v_t}} = \sqrt{\frac{y_0^2 \bar\rho^2 v_t}{2\pi}} \int_{\mathbb{R}} e^{iu\xi} e^{-\frac{y_0^2 \bar\rho^2 v_t \xi^2}2} d\xi
\eeaa
and letting $u = x - x_0 - \rho y_0 \int_0^t e^{\nu B^H_s} dB_s + \frac{y_0^2v_t}2 $, we have
\bea\label{eqn:Fourier}
&& \frac1{\sqrt{2\pi y_0^2 \bar\rho^2 v_t}} e^{-\frac1{2 y_0^2 \bar\rho^2 v_t}{\left(x- x_0 - y_0 \rho \int_0^t e^{\nu B^H_s} dB_s + \frac{y_0^2v_t }2 \right)^2}} \nonumber\\
&=& \frac1{2\pi} \int e^{i\left(x- x_0 - \rho y_0 \int_0^t e^{\nu B^H_s} dB_s + \frac{y_0^2 v_t}2 \right)\xi} e^{-\frac{y_0^2 \bar\rho^2 v_t \xi^2}2} d\xi.
\eea
Plugging \eqref{eqn:Fourier} into the right-hand side of \eqref{eqn:density-cont}, we get
\beaa
&& \Eof{f(X_t,Y_t)} \nonumber\\
&=&  \frac{1}{y\sqrt{2\pi \nu^2t^{2H}}}\frac1{2\pi}\int_{\mathbb{R}^2}f(x,y)e^{-\frac{\left(\ln{(y/y_0)}\right)^2}{2\nu^2t^{2H}}} \times\nonumber\\
&&\ \int_{\mathbb{R}}\Eof{\left.e^{i\left(x- x_0 - \rho y_0 \int_0^t e^{\nu B^H_s} dB_s + \frac{y_0^2v_t}2  \right)\xi} e^{-\frac{y_0^2 \bar\rho^2 v_t \xi^2}2} \right|B_t^H=\frac{\ln{(y/y_0)}}{\nu}}d\xi\,dx\,dy.
\eeaa
Finally, we end up with the following bridge representation of the density \eqref{eqn:bridge-rep}.
\end{proof}

By transforming back to the original variables $(s, a) = ( e^{x}, y)$, we obtain a bridge representation for the joint density $q(t;s,a)$ of $(S_t, \alpha_t)$ in \eqref{eqn:fSABR-s}.  
\begin{cor}
The joint density $q(t;s,a)$ of the lognormal fractional SABR model \eqref{eqn:fSABR-s} has the following bridge representation
\bea
  q(t;s, a) \label{eqn:bridge-rep-s}   &=& \frac{e^{-\frac{(\ln( a/a_0))^2}{2\nu^2 t^{2H}}}}{ a\sqrt{2\pi \nu^2 t^{2H}}}\, \frac1{2\pi s}\\
&&\times \int_{\mathbb{R}} \left(\frac{s}{s_0}\right)^{i\xi} \Eof{\left. e^{i\left(-\rho \int_0^t a_0 e^{\nu B^H_s} dB_s + \frac{a_0^2v_t}2  \right)\xi} e^{-\frac{\bar\rho^2 a_0^2 v_t}2 \xi^2}\right|B^H_t=\frac{\ln (a/a_0)}{\nu}} d\xi. \nonumber
\eea
\end{cor}

%
%

\section{Expansion around deterministic path} \label{sec:density-exp}
To gain more intuition and in particular a more practical form for applications in obtaining approximations of implied volatility, this section is devoted to deriving an expansion to the lowest order of the bridge representation \eqref{eqn:bridge-rep} around a properly chosen deterministic path.
The expansion will be shown useful in deriving a small time approximation for implied volatility in Section \ref{sec:implied-vol}.

Recall that the joint density $p$ of $(X_t, Y_t)$ has the representation given in \eqref{eqn:bridge-rep} as
\beaa
  && p(t;x,y) \nonumber\\
  &=& \frac1{y\sqrt{2\pi \nu^2 t^{2H}}} e^{-\frac{\left(\ln{(y/y_0)}\right)^2}{2\nu^2 t^{2H}}} \times \nonumber\\
  && \frac1{2\pi} \int_{\mathbb{R}} \Eof{\left. e^{i\left(x - x_0 - \rho y_0 \int_0^t e^{\nu B^H_s} dB_s + \frac{y_0^2v_t}2  \right)\xi} e^{-\frac{\bar\rho^2 y_0^2 v_t\xi^2}2 }\right|B_t^H=\frac{\ln{(y/y_0)}}{\nu}} d\xi.\eeaa
Let us start with a few na\"ive calculations as follows. We expand the above conditional expectation around the deterministic path $m_s$, for $0 \leq s \leq t$, that is determined by the conditional expectation of $B_s^H$ given its terminal point $B_t^H=\frac{\ln (y/y_0)}{\nu}$. Precisely,
\[
m_s := \mathbb{E}\left[B^H_s\left|B_t^H=\frac{\ln{(y/y_0)}}{\nu} \right.\right] = R\left(1, \frac st\right)\frac{\ln{(y/y_0)}}\nu,
\]
where $R$ is defined in \eqref{eqn:Rts}. By Taylor's expansion,  we have, for $n \geq 0$,
\beaa
&& e^{-i\rho\xi \int_0^t y_0 e^{\nu B^H_s} dB_s} e^{-\frac 12 (\bar\rho^2\xi - i) \xi \int_0^t y_0^2 e^{2 \nu B^H_s} ds} \\
&\approx& e^{-i\rho\xi \int_0^t y_0 e^{\nu m_s} dB_s} e^{-\frac12 (\bar\rho^2\xi - i) \xi \int_0^t y_0^2 e^{2 \nu m_s} ds} \times \\
&& \sum_{k, \ell=0}^n \frac{(- i\rho\xi)^k}{k!} \left\{ \int_0^t y_0 \left(e^{\nu B^H_s} - e^{\nu m_s}\right) dB_s \right\}^k
\times \\
&&
\frac1{\ell!}\left\{-\frac12 (\bar\rho^2\xi - i) \xi \int_0^t y_0^2 \left( e^{2 \nu B^H_s} - e^{2 \nu m_s} \right)ds \right\}^\ell .
\eeaa
Thus, even for obtaining a na\"ive expansion, we shall need a systematic way of computing the conditional expectations of the form, for either $k\geq 1$ or $\ell \geq 1$,
\[
\Eof{\left. e^{-i\rho\xi \int_0^t y_0 e^{\nu m_s} dB_s} \left\{ \int_0^t \left(e^{\nu B^H_s} - e^{\nu m_s}\right) dB_s \right\}^k \left\{\int_0^t \left( e^{2 \nu B^H_s} - e^{2 \nu m_s} \right)ds \right\}^\ell\right|B_t^H=\frac{\ln{(y/y_0)}}{\nu}},
\]
which is pretty complicated if not impossible. Nevertheless, as leading order is concerned, small time expansion of the joint density $p$ to the lowest order (i.e., $k=\ell=0$) is still manageable. The result is summarized in the following theorem.

In the following sequel, for simplification of the notation, we use $\E_{\frac\eta\nu}[\cdot]$ to denote $\Eof{\left.\cdot\right|B_t^H=\frac\eta\nu}$, where $\eta=\ln(y/y_0)$. A function $g$ is denoted by $g(t) = O(t^a)$ as $t \to 0^+$ if it satisfies
\[
 \limsup_{t\to 0^+} \frac{|g(t)|}{t^a}< \infty.
\]
\begin{theorem} \label{thm:exp-rho-nzero}
The joint probability density $p$ of the process $(X_t, Y_t)$ satisfying \eqref{eqn:fSABR-x} has the following asymptotic to the lowest order
\bea\label{eqn:small-time-p}
&& p(t;x,y)  \\
&=& \frac1{2\pi} \, \frac1{y \sqrt{\nu^2 t^{2H}}} e^{-\frac{\eta^2}{2 \nu^2t^{2H}}} \, \frac1{y_0\sqrt{t \psi(\eta)}} e^{-\frac1{2y_0^2\psi(\eta)}\left(\frac{x - x_0}{\sqrt t} +\frac{y_0^2\sqrt{t}}2C_{eR}(\eta)- \rho y_0t^{- H} C_{RK}(\eta) \frac\eta\nu  \right)^2} \left(1 + O\left(\sqrt t \right) \right), \nonumber
\eea
where
\beaa
&& C_{RK}(\eta) := \int_0^1 e^{R(1,u)\frac\eta\nu} K_H(1,u) du, \label{eqn:CRK} \\
&& C_{eR}(\eta) := \int_0^1 e^{2 R(1,u)\frac\eta\nu} du, \label{eqn:CeR} \\
&& \psi(\eta):= C_{eR}(\eta) - \rho^2 C_{RK}^2(\eta).\label{eqn:vtilde}
\eeaa
\end{theorem}
\begin{proof} 
To the lowest order, $p$ is given by
\bea\label{eqn:bridge-rep-0th} 
  && p(t;x,y)\nonumber \\
  &=& \frac{e^{-\frac{\eta^2}{2 \nu^2t^{2H}}}}{ y\sqrt{2\pi \nu^2 t^{2H}}}\, \frac1{2\pi} \int_\mathbb{R} e^{i(x- x_0)\xi} e^{-\frac12 (\bar\rho^2\xi - i) \xi \int_0^t y_0^2 e^{2 \nu m_s} ds} \E_{\frac\eta\nu}\left[e^{-i\rho\xi \int_0^t y_0 e^{\nu m_s} dB_s} \right] d\xi. 
\eea
We consider the conditional expectation in the above expression.
Note that $\int_0^t e^{\nu m_s} dB_s$ and $B_t^H$ are jointly Gaussian. We  apply the following identity to evaluate the conditional expectation: if $X$ and $Y$ are jointly normal with mean 0, we can decompose $X$ as
\[
X = \frac{\cov(X,Y)}{V(Y)} Y + \sqrt{\frac{V(X)V(Y) - \cov(X,Y)^2}{V(Y)}}Z,
\]
where $Y$ and $Z$ are independent and $Z$ is standard normal. Hence,
\beaa
\Eof{f(X)|Y = y} = \Eof{f\left( \frac{\cov(X,Y)}{V(Y)} y + \sqrt{\frac{V(X)V(Y) - \cov(X,Y)^2}{V(Y)}}Z \right)}.
\eeaa
In our case, $X = \int_0^t e^{\nu m_s} dB_s$ and $Y = B_t^H$, hence
\beaa
&& V(X) = \int_0^t e^{2\nu m_s} ds = t \, \int_0^1 e^{2R(1,u)\eta} du = C_{eR}(\eta) \, t, \\
&& V(Y) = t^{2H}, \\
&& \cov(X,Y) = \int_0^t e^{\nu m_s} K_H(t,s) ds = t^{H + \frac12} \, \int_0^1 e^{R(1,u)\eta} K_H(1,u) du = C_{RK}(\eta)\, t^{H + \frac12}.
\eeaa
Therefore,
\beaa
&& \E_{\frac\eta\nu} \left[ e^{-i\rho\xi \int_0^t y_0 e^{\nu m_s} dB_s} \right]\nonumber \\
&=& e^{-i\rho\xi y_0 t^{\frac12 - H}\, C_{RK}(\eta)\frac\eta\nu} \E\left[ e^{-i\rho\xi y_0 \left\{\sqrt t \, \sqrt{C_{eR}(\eta) - C_{RK}^2(\eta)} \right\}Z} \right] \nonumber\\
&=& \exp\left[-i\rho\xi y_0 t^{\frac12 - H} C_{RK}(\eta)\frac\eta\nu - \frac{\rho^2\xi^2y_0^2 t }2 \left\{C_{eR}(\eta) - C_{RK}^2(\eta) \right\}\right].
\eeaa


Thus, by substituting the above expression into \eqref{eqn:bridge-rep-0th}, we obtain
\bea\label{sep-s-3-6}
&& p(t;x,y)\nonumber \\
&=&\frac{1}{2\pi} \frac1{y \sqrt{\nu^2 t^{2H}}} e^{-\frac{\eta^2}{2 \nu^2 t^{2H}}} \times\nonumber \\
&& \frac1{\sqrt{2\pi}}\int_\mathbb{R} e^{i(x - x_0)\xi} e^{-\frac12 (\bar\rho^2\xi - i) \xi t y_0^2 C_{eR}(\eta)}\,  e^{-i\rho\xi y_0 t^{\frac12 - H} C_{RK}(\eta) \frac\eta\nu - \frac{\rho^2\xi^2y_0^2 t }2 \left\{C_{eR}(\eta) - C_{RK}^2(\eta) \right\}}
\, d\xi \nonumber\\
&=&\frac{1}{2\pi}  \frac1{y\sqrt{ \nu^2 t^{2H}}} e^{-\frac{\eta^2}{2\nu^2  t^{2H}}} \times \nonumber\\
&& \frac1{\sqrt{2\pi}} \int_\mathbb{R} e^{i\left(x- x_0 + \frac{y_0^2t}2 C_{eR}(\eta)- \rho y_0 t^{\frac12 - H} C_{RK}(\eta)\frac\eta\nu   \right)\xi} \, e^{-\frac{\xi^2y_0^2t}2  \left\{C_{eR}(\eta) - \rho^2 C_{RK}^2(\eta) \right\}}
 \, d\xi\nonumber \\
&=& \frac1{2\pi} \, \frac1{y \sqrt{\nu^2 t^{2H}}} e^{-\frac{\eta^2}{2 \nu^2t^{2H}}} \, \frac1{y_0\sqrt{t \psi(\eta)}} e^{-\frac1{2y_0^2\psi(\eta)}\left(\frac{x - x_0}{\sqrt t} +\frac{y_0^2\sqrt{t}}2C_{eR}(\eta)- \rho y_0t^{- H} C_{RK}(\eta) \frac\eta\nu  \right)^2}.
\eea
We postpone the detailed error analysis to Section \ref{sec:error-estimate} in the appendix.
\end{proof}
\begin{remark}
We remark that in the logarithmic scale, \eqref{eqn:small-time-p} can be expressed in a more concise form as
\beaa
&& \ln p(t; x, y) \\
&=&-\frac1{2 t^{2H}}\left[\frac{\eta^2}{\nu^2} + \frac1{y_0^2 \psi(\eta)}\left(\frac{x - x_0}{t^{\frac12 - H}}+\frac{y_0^2t^{H+\frac12}}2C_{eR}(\eta) - \rho y_0 C_{RK}(\eta) \frac\eta\nu \right)^2 \right]+O(\ln t).
\eeaa
\end{remark}

\begin{remark}
In the case that $\nu = 1$, $\rho = 0$, and $H = \frac12$, we have
\[
C_{eR}(\eta) = \int_0^1 e^{2R(1,u)\eta} du = \int_0^1 e^{2u\eta} du = \frac1{2\eta}(e^{2\eta} - 1) = \frac{y^2 - y_0^2}{2\eta y_0^2}.
\]
Then \eqref{eqn:small-time-p} reduces to
\bea
&& \frac1{2\pi} \times \frac1{y\sqrt{t}} e^{-\frac{\eta^2}{2 t}} \times \frac1{y_0 \sqrt{C_{eR}(\eta)t}} e^{-\frac1{2y_0^2 C_{eR}(\eta) t  }\left(x - x_0 + \frac{y_0^2t}2 C_{eR}(\eta)  \right)^2} \nonumber \\
&=& \frac1{2\pi} \, \frac1{y \sqrt{t}} e^{-\frac{\eta^2}{2 t}} \,
\frac1{y_0 \sqrt{C_{eR}(\eta)t}}
e^{-\frac{\left(x - x_0\right)^2}{2y_0^2 C_{eR}(\eta) t}}
e^{-\frac{x - x_0}{2}} (1 + O(t)) \nonumber \\
&=& \frac1{2\pi t} \, e^{-\frac1{2t}\left[\eta^2 + \frac{2\eta \left(x - x_0\right)^2}{y^2 - y_0^2}\right]}
\frac{e^{-\frac{x - x_0}{2}}}{y y_0 \sqrt{C_{eR}(\eta)}} (1 + O(t)). \label{eqn:hyperbolic-hk-approx}
\eea
Notice that in this case $(X_t, Y_t)$ represents the Brownian motion in hyperbolic plane whose transition density $p_\mathbb{H}$ (with respect to the Riemannian area measure) has the leading term in small time asymptotic as
\[
p_\mathbb{H}(t;x, y) = \frac1{2\pi t} e^{-\frac{d^2(x,y;x_0,y_0)}{2t}} (1 + O(t)),
\]
where $d$ denotes the geodesic distance between $(x,y)$ and $(x_0,y_0)$ in the hyperbolic plane. For reader's reference, the hyperbolic cosine of the geodesic distance $d(x,y;x_0,y_0)$ has the closed form expression
\[
\cosh d(x, y; x_0, y_0) = \frac{(x - x_0)^2 + y_0^2 + y^2}{2 y_0 y}.
\]
Thus, in a sense the following function in \eqref{eqn:hyperbolic-hk-approx}
\[
\tilde d(x,y; x_0, y_0)
:= \sqrt{\eta^2 + \frac{2 \eta}{y^2 - y_0^2} \, (x - x_0)^2}
\]
can be regarded as an approximation of the hyperbolic geodesic distance. The complete recovery of the hyperbolic geodesic distance is demonstrated in Section \ref{sec:ldp} below.
\end{remark}
\section{Small time approximation of option price and implied volatility} \label{sec:implied-vol}

We derive in this section the small time asymptotics of the premium of a call option and its associated implied volatility by applying the small time asymptotics for the probability density obtained in Section \ref{sec:density-exp} when $H \leq \frac12$.
It is documented, for exmaple, in Ekstr\"om and Lu \cite{ekstrom-lu}, that if the underlying asset is governed by an exponential L\'evy model, the induced implied volatilities of non ATM options may explode if jumps exist and the underlying process jumps towards the strike. As we shall see in the following, when $H < \frac12$, the small time approximation of implied volatility also explodes; creating a jump like behavior in the underlying process.

Let $k = \ln K$ be the logmoneyness, $t$ the time to expiry, and recall that $S_t =  e^{X_t}$. Though equivalently, we shall be primarily working with the $(X_t, Y_t)$ process as in \eqref{eqn:fSABR-x} rather than the $(S_t, \alpha_t)$ process in \eqref{eqn:fSABR-s} hereafter. We write the price $C$ of a call as a function of $k$ and $t$ as
\beaa
C(k,t) &:=& \Eof{\left(S_t - K\right)^+} =  \Eof{(e^{X_t} - e^k)^+} \\
&=&  \iint (e^{x} - e^k)^+ p(t; x, y) dx\, dy.
\eeaa
To evaluate the last integral, we approximate the joint density $p$ by the small time asymptotics obtained in Theorem \ref{thm:exp-rho-nzero}, then, as $t \to 0^+$, apply Laplace asymptotic formula to the resulting integral. For reader's convenience, we provide with proof in Secton \ref{sec:laplace-formula} a variation of the Laplace asymtotic formula that is tailored for our own use.


\begin{lemma} \label{lma:logC-small-time} Let $H\leq \frac12$.
For out-of-money call options, i.e., $k > x_0$, the call price $C(k,t)$ has the following asymptotic as $t \to 0$
\bea
\ln C(k,t) \approx -\frac1{2t^{2H}} \left\{\frac{\eta_*^2}{\nu^2}+ \frac1{y_0^2 \psi(\eta_*)}\left(\frac {k-x_0}{t^{\frac12 - H}} - \rho y_0 C_{RK}(\eta_*) \frac{\eta_*}{\nu} \right)^2\right\}, \label{eqn:logC-asymp}
\eea
where $\eta_*$ is the minimizer
\[
\eta_* = \argmin\left\{ \eta \in \R: \frac{\eta^2}{\nu^2} + \frac1{y_0^2 \psi(\eta)}\left(\frac{ k-x_0}{t^{\frac12 - H}} - \rho y_0 C_{RK}(\eta)\frac\eta\nu\right)^2 \right\}.
\]
\end{lemma}
\begin{proof}
The proof is a straigtforward application of the Laplace asymptotic formula \eqref{eqn:lap-int-1} in Lemma \ref{lma:Laplace-asymp}.  Let $\mathcal{C} = \{(x,\eta):x \geq k\}\subseteq\mathbb{R}^2$ and $\alpha = \frac12 - H \geq 0$. By using the asymptotic density \eqref{eqn:small-time-p}, consider
\beaa
 C(k,t)& =&  \int_0^\infty\int_k^\infty (e^{x} - e^k) p(t;x,y) dx dy \\
&=& \frac{1}{2\pi} \int_0^\infty\int_k^\infty \left(e^{x} - e^k \right) \left\{
\frac1{y \sqrt{\nu^2 t^{2H}}} e^{-\frac{\eta^2}{2 \nu^2t^{2H}}} \,
\frac1{y_0 \sqrt{t\psi(\eta)}} \right. \times \\
&& \left.e^{-\frac1{2y_0^2t\psi(\eta)}\left(x - x_0 - \rho y_0 C_{RK}(\eta) \frac\eta\nu t^{\frac12 - H} \right)^2} 
 e^{-\frac{x - x_0}{2\psi(\eta)}C_{eR}(\eta)} \left( 1 + O\left(\sqrt t \right)\right) \right\} dx dy \\
&=& \frac{1}{2\pi \nu y_0 t^{H + \frac12}} \iint\limits_{\mathcal{C}} \left(\frac{e^{x} - e^k}{\sqrt{\psi(\eta)}} \right) \,
e^{-\frac{x- x_0}{2\psi(\eta)}C_{eR}(\eta)} \times \\
&& \qquad e^{-\frac1{2t} \left\{\frac{\eta^2}{\nu^2}  t^{2\alpha} + \frac1{y_0^2 \psi(\eta)}\left(x - x_0 - \rho y_0 C_{RK}(\eta) \frac\eta\nu t^\alpha \right)^2 \right\}} \, \left( 1 + O\left(\sqrt t \right)\right) dx d\eta.
\eeaa
Applying the Laplace asymptotic formula \eqref{eqn:lap-int-1} to the lowest order term in the last expresion yields
\beaa
- \ln C(k,t) &\approx& \frac1{2t} \left\{\frac{\eta_*^2}{\nu^2}  t^{2\alpha} + \frac1{y_0^2 \psi(\eta_*)}\left(x_* - x_0 - \rho y_0 C_{RK}(\eta_*) \frac{\eta_*}{\nu} t^\alpha \right)^2 \right\} \\
&=& \frac1{2t^{2H}} \left\{\frac{\eta_*^2}{\nu^2}+ \frac1{y_0^2 \psi(\eta_*)}\left(\frac {x_*- x_0}{t^\alpha}  - \rho y_0 C_{RK}(\eta_*) \frac{\eta_*}{\nu} \right)^2 \right\},
\eeaa
where, for fixed $t$, $(x_*,\eta_*)$ is the minimizer of the function
\beaa
(x_*,\eta_*) &=& \argmin\left\{(x,\eta)\in\mathcal{C}: \frac{\eta^2}{\nu^2}  t^{2\alpha} + \frac1{y_0^2 \psi(\eta)}\left(x - x_0 - \rho y_0 C_{RK}(\eta)\frac\eta\nu t^\alpha \right)^2 \right\} \\
&=& \argmin\left\{(x,\eta)\in\mathcal{C}: \frac{\eta^2}{\nu^2} + \frac1{y_0^2 \psi(\eta)}\left(\frac{ x-x_0}{t^\alpha} - \rho y_0 C_{RK}(\eta) \frac\eta\nu \right)^2 \right\}.
\eeaa
Since the objective function is continuous in $(x,\eta)\in\mathcal{C}$ and it is a quadratic function in $x$,  it follows that $x_* = k$ when $t$ is small enough, thereby
\[
\eta_* = \argmin\left\{\eta : \frac{\eta^2}{\nu^2} + \frac1{y_0^2 \psi(\eta)}\left(\frac {k-x_0}{t^\alpha} - \rho y_0 C_{RK}(\eta)\frac\eta\nu \right)^2 \right\}.
\]
\end{proof}
\begin{remark} 
For the case $H>\frac12$, from asymptotic density \eqref{eqn:small-time-p} it implies
\beaa
 C(k,t)& =&  \int_0^\infty\int_k^\infty (e^{x} - e^k) p(t;x,y) dx dy \\
&=& \frac{1}{2\pi \nu y_0 t^{H + \frac12}} \iint\limits_{\mathcal{C}} \left(\frac{e^{x} - e^k}{\sqrt{\psi(\eta)}} \right) 
e^{-\frac1{2t^{2H}}\left[\frac{\eta^2}{\nu^2}+\left(\frac{(x-x_0)t^{H-\frac12}}{y_0\sqrt{\psi(\eta)}}+\frac{y_0C_{eR}(\eta)t^{H+\frac12}}{2\sqrt{\psi(\eta)}}-\frac{\rho C_{RK}(\eta)\eta}{\sqrt{\psi(\eta)}\nu}\right)^2\right]}.
\eeaa
When $t$ is small enough, the minimum value of 
\[
\frac{\eta^2}{\nu^2}+\left(\frac{(x-x_0)t^{H-\frac12}}{y_0\sqrt{\psi(\eta)}}+\frac{y_0C_{eR}(\eta)t^{H+\frac12}}{2\sqrt{\psi(\eta)}}-\frac{\rho C_{RK}(\eta)\eta}{\sqrt{\psi(\eta)}\nu}\right)^2
\]
is not attained at the boundary of $\mathcal{C}$. Hence, the Laplace asymptotic formula cannot be applied in this case. However, when $H>\frac12$, we can still present the following uniqueness of the minimal point $\eta^\ast$ graphically in the following Remark.
\end{remark}
\begin{remark}
The plots in Figure \ref{fig:contour} shows graphically the uniqueness of the minimal point $\eta_*$ for $H = \frac14$ and $H = \frac34$. In theses particular examples, the contours are convex in the half plane $x > k$, which corresponds to the out-of-money calls. For out-of-money puts, $x < k$, though the contours are not convex, the uniqueness of $\eta_*$ sustains.
\begin{figure}[ht!]
\begin{center}
\includegraphics[width=7cm, height=8cm]{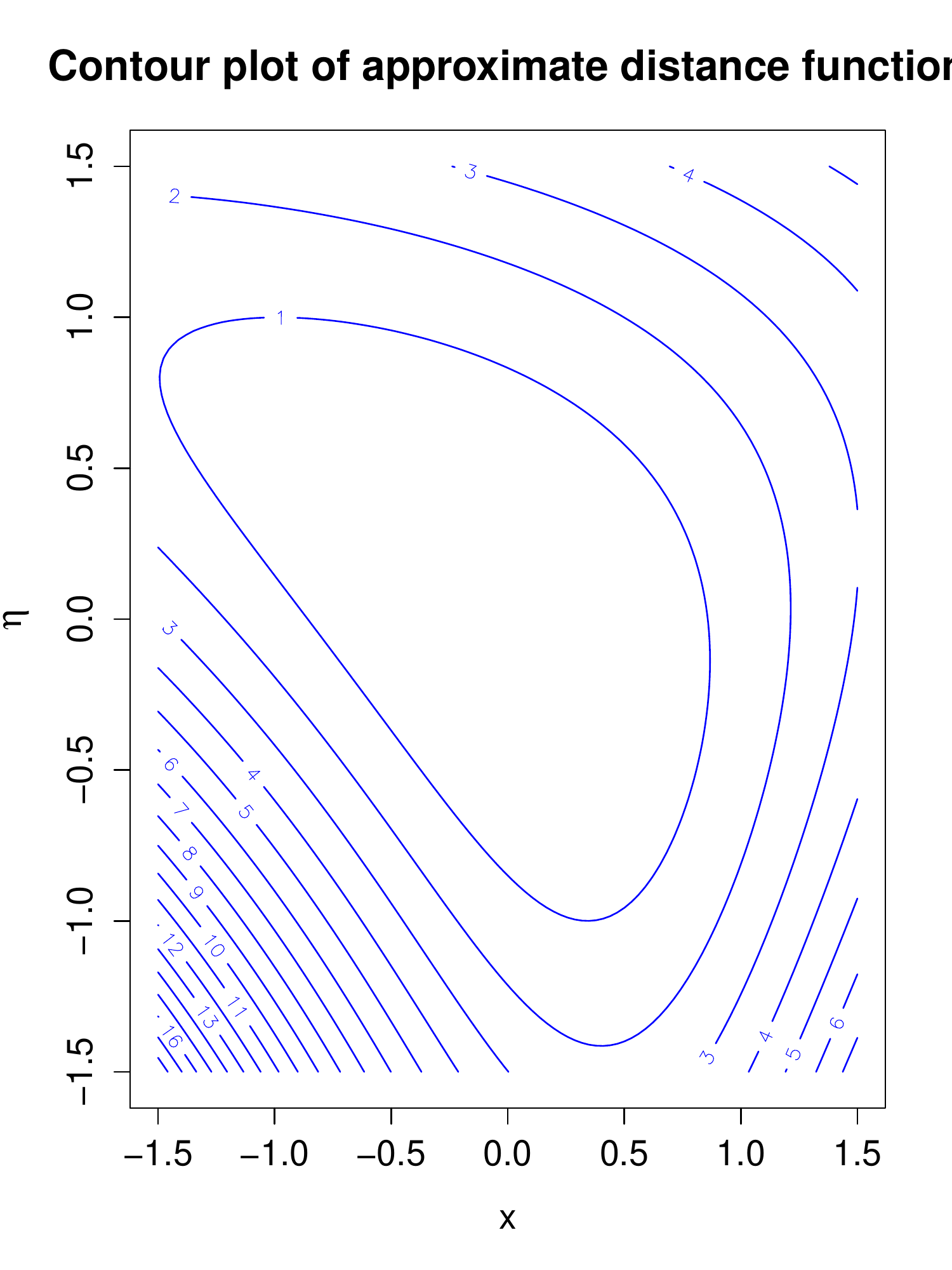} \;
\includegraphics[width=7cm, height=8cm]{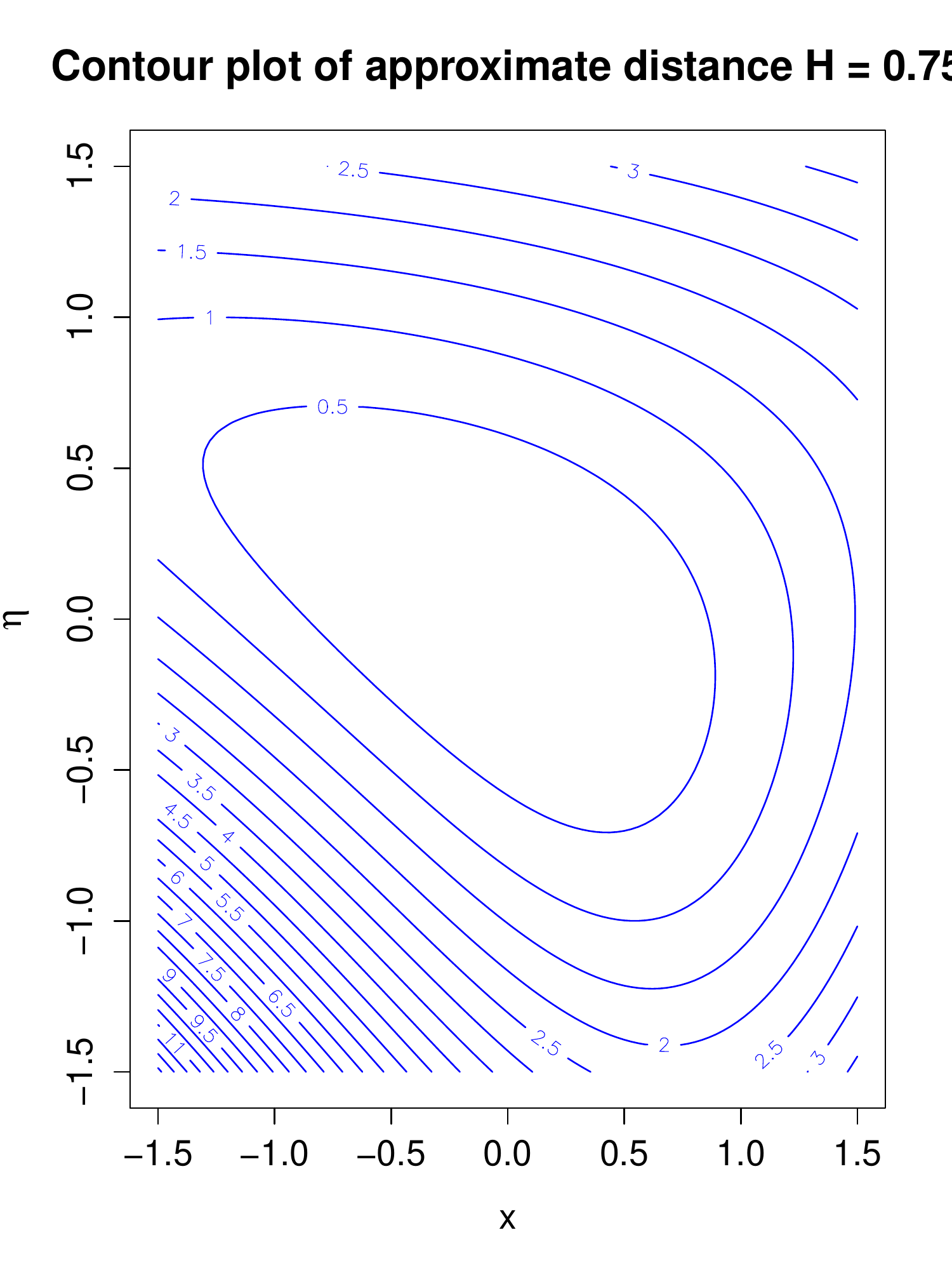}
\end{center}
\caption{The contour plots. Parameters $\rho = -0.7$, $\nu=1$, $y_0 = 1$, $t = 0.5$. $H = 0.75$ on the right; $H = 0.25$, on the left.}
\label{fig:contour}
\end{figure}
\end{remark}
So long as we establish an asymptotic for the log price $\ln C(k,t)$ for $k > x_0$, by using the following small time asymptotic for implied volatility in \cite{gao-lee} or \cite{roper-rutkowski}
\begin{equation}
\sigma_{\BS}(k,t) = \frac{|k-x_0|}{\sqrt{2 t |\ln C(k,t)|}} + O\left( \frac{\ln|\ln C(k,t)|}{\sqrt t |\ln C(k,t)|^{3/2}} \right), \label{eqn:iv-small-time-asymp}
\end{equation}
an asymptotic formula for implied volatility follows immediate. We summarize the result in the following theorem but omitting its proof.
\begin{theorem}  \label{thm:fSABR-fake} 
Let $H\leq\frac12$ and let $k = \ln K$ be the log moneyness and $\alpha = \frac12 - H$.
The implied volatility $\sigma_{\BS}(k,t)$ for out-of-money calls ($k > x_0$) has the following asymptotic in small time to expiry
\begin{equation}
\sigma_{\BS}^2(k,t) = \sigma_{\BS}^2\left(\frac k{t^\alpha} \right) \approx \frac{(k-x_0)^2}{t^{2\alpha}} \left\{\frac{\eta_*^2}{\nu^2} + \frac1{y_0^2 \psi(\eta_*)} \left(\frac {k-x_0}{t^\alpha} - \rho y_0^2 C_{RK}(\eta_*)\frac{\eta_*}{\nu} \right)^2 \right\}^{-1}. \label{eqn:fSABR-fake}
\end{equation}
The minimal point $\eta_*$ is given Lemma in \ref{lma:logC-small-time}.
\end{theorem}

\begin{remark}
Note that \eqref{eqn:fSABR-fake} does not recover the SABR formula when $H = \frac12$. The derivation of the SABR formula relies heavily on the geometry and symmetry of the underlying SABR plane which is isometric to the Poincar\'e plane. Figure \ref{fig:SABR-vs-fSABR} shows the comparison between the two formulas with time to expiry $t=1$. Parameters are chosen so as to reproduce the figures in \cite{hklw}. In this set of parameters, the maximal difference between the two approximate implied volatility curves is about 1\% for logmoneyness $k \in [-1,1]$.
\end{remark}
\begin{figure}[ht!]
\begin{center}
\includegraphics[width=7cm, height=8cm]{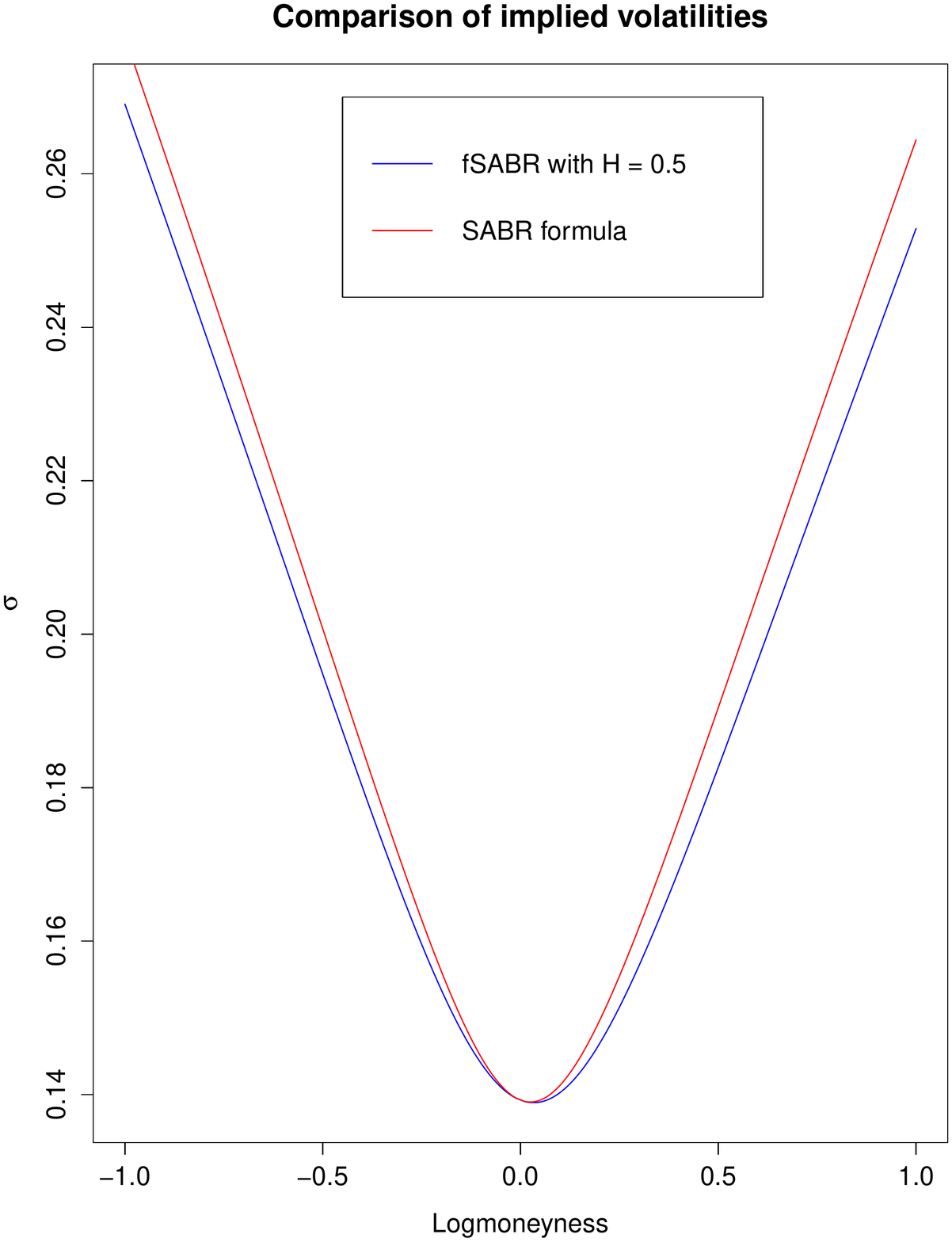} \;
\includegraphics[width=7cm, height=8cm]{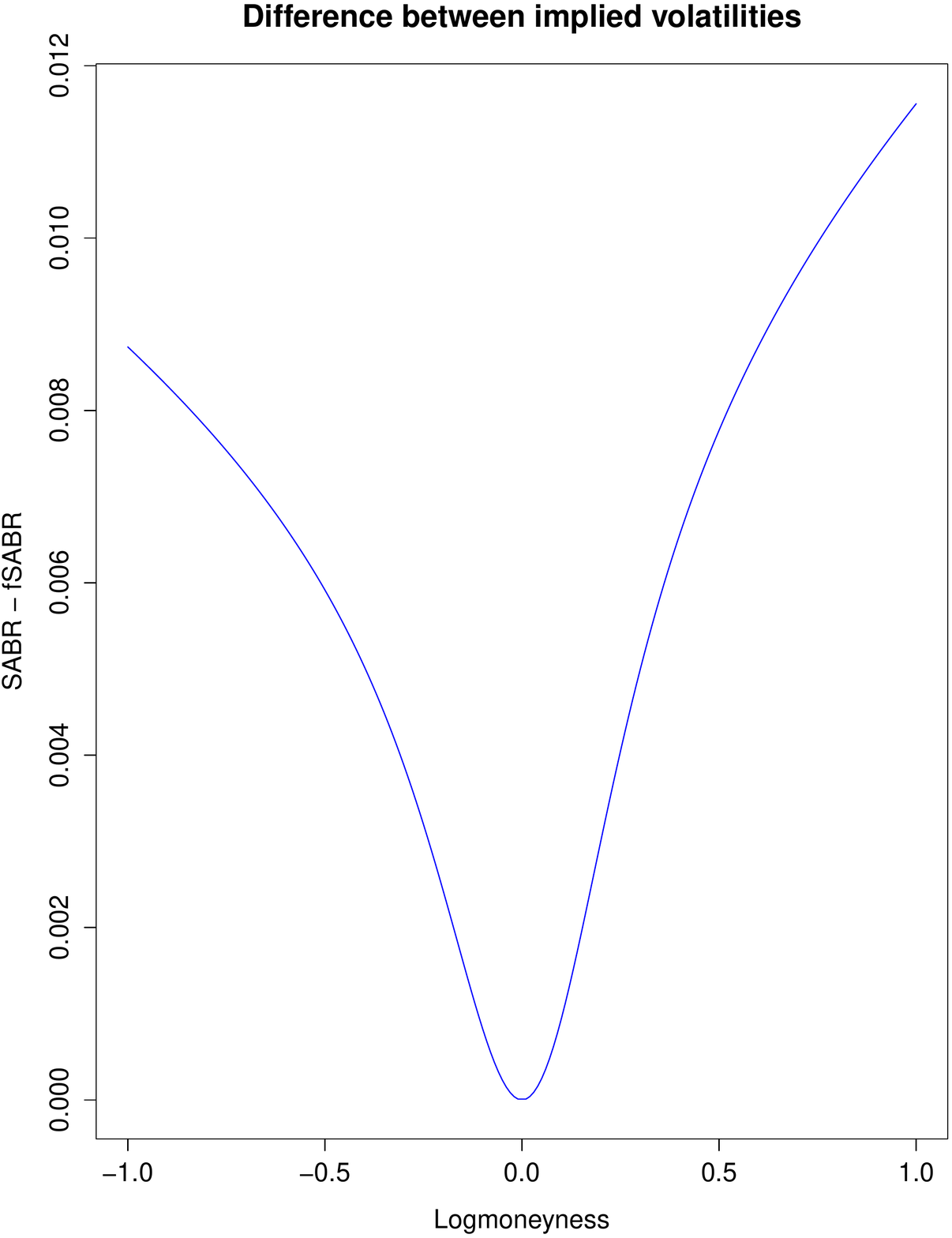}
\end{center}
\caption{The plot on the left shows the approximate implied volatility curves versus logmoneyness with time to expiry $t = 1$ produced by \eqref{eqn:fSABR-fake} (in blue) the SABR formula \eqref{eqn:sabr-formula} (in red). Parameters are set as $\rho = -0.06867 $, $\nu = 0.5778$, $a_0 = 0.13927$. The plot on the right shows the difference between the two curves.}
\label{fig:SABR-vs-fSABR}
\end{figure}

 We conclude the section by remarking that, as time to expiry $t$ approaches zero, the approximate 
implied volatility $\sigma_{\BS}(k,t)$ flattens out with $H > \frac12$;
whereas the whole surface $\sigma_{\BS}(k,t)$ explodes with $H < \frac12$ except for the at-the-money option $k=x_0$. Figure \ref{fig:imp-vols} shows the plots of approximate implied volatilities $\sigma$ given in \eqref{eqn:fSABR-fake}
versus logmoneyness $k$ for time to expiry $t = 0.01$ and $t=1$ respectively, and various Hurst exponenets $H$. As in Figure \ref{fig:SABR-vs-fSABR}, parameters are chosen as $a_0 = 0.13927$, $\nu = 0.5778$, and $\rho = -0.06867$. The numerical determination of the $\eta_*$'s is relatively efficient since it is basically a one-dimensional optimization problem.
\begin{figure}[ht!]
\begin{center}
\includegraphics[width=7cm, height=8cm]{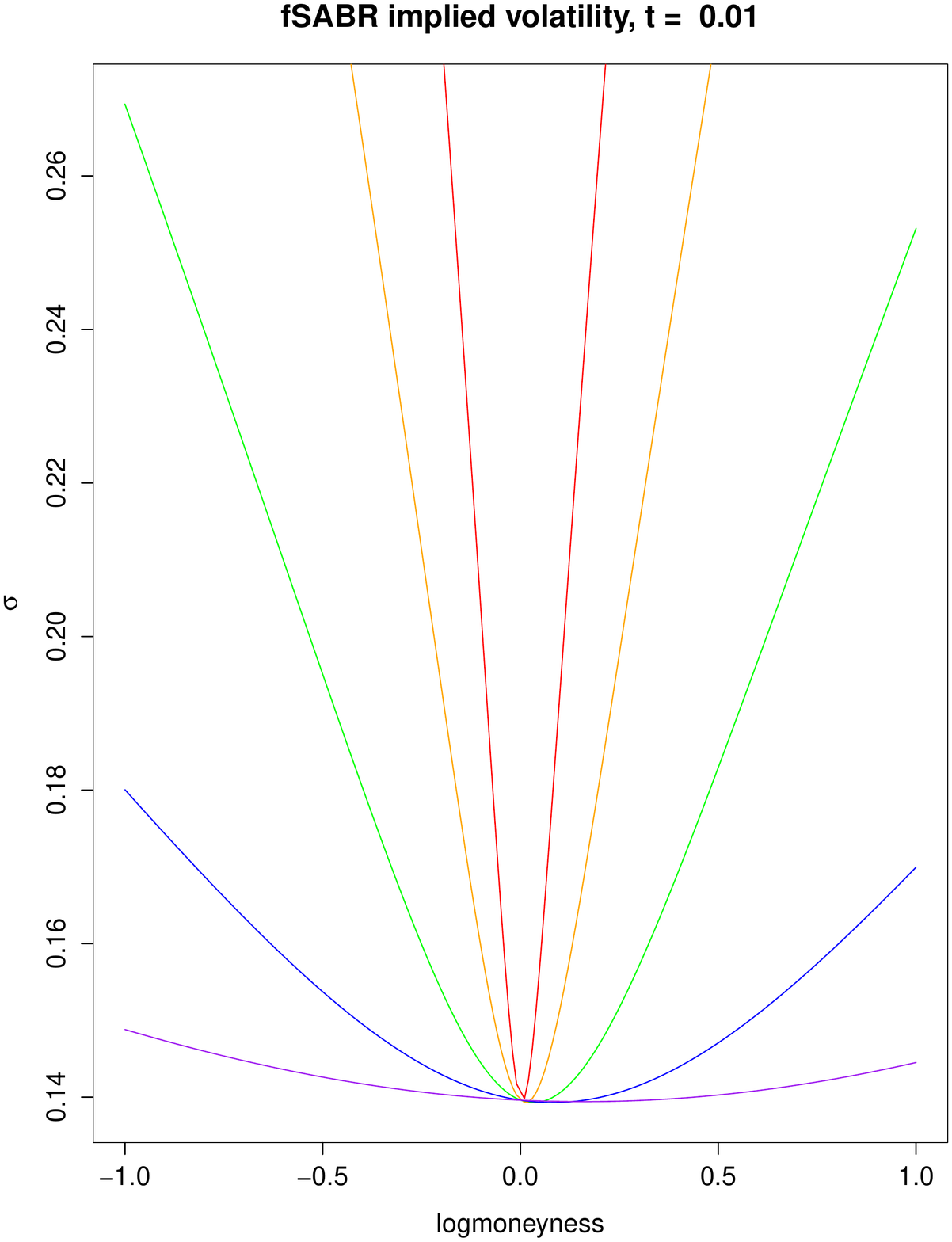} \;
\includegraphics[width=7cm, height=8cm]{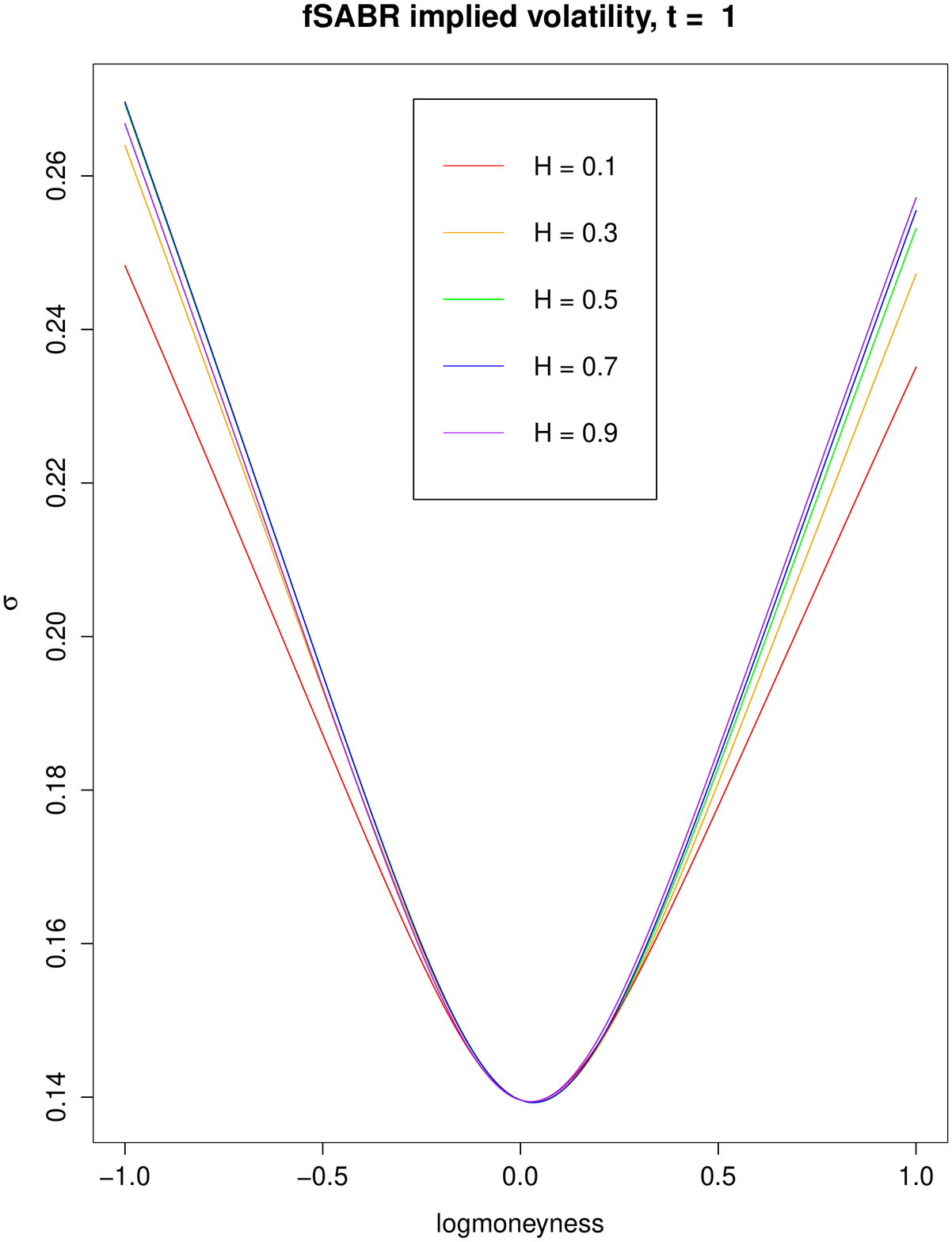}
\end{center}
\caption{The implied volatility curves for $t = 0.01$ on the left, $t = 1$ on the right. Parameters are set as $\rho = -0.06867 $, $\nu = 0.5778$, $a_0 = 0.13927$. $H=0.1$ in red, $H=0.3$ in orange, $H = \frac12$ in green, $H=0.7$ in blue, $H=0.9$ in purple.}
\label{fig:imp-vols}
\end{figure}


%
%

\section{A heuristic large deviation principle} \label{sec:ldp}
In this section, we provide a heuristic derivation of the sample path large deviation principle for $(X_t, Y_t)$ in small time by bootstrapping the bridge representation to multiperiod. For notational simplicity, we introduce the following vector notations.
\beaa
&& \bt = (t_1, \cdots, t_n), \qquad \bx_{\bt} = (x_{t_1}, \cdots, x_{t_n}), \qquad \by_{\bt} = (y_{t_1}, \cdots, y_{t_n}) ,\\
&& \bB^H_{\bt} = (B^H_{t_1}, \cdots, B^H_{t_n}), \qquad \bX_{\bt} = (X_{t_1}, \cdots, X_{t_n}), \qquad \bY_{\bt} = (Y_{t_1}, \cdots, Y_{t_n}), \\
&& \bsxi_{\bt} = (\xi_{t_1}, \cdots, \xi_{t_n}), \qquad \bseta_{\bt} = (\eta_{t_1}, \cdots, \eta_{t_n}), \qquad \bszeta_{\bt} = (\zeta_{t_1}, \cdots, \zeta_{t_n}).
\eeaa
\begin{theorem} \label{thm:rep-multi}
The multiperiod joint density $p$ of $(X_t, Y_t)$
\beaa
&& p(x_1, y_1, \cdots, x_n, y_n) := \Pof{(X_{t_1}, Y_{t_1}) = (x_1, y_1), \cdots, (X_{t_n}, Y_{t_n}) = (x_n, y_n)}
\eeaa
has the following bridge representation
\bea
&& p(x_1, y_1, \cdots, x_n, y_n) \label{eqn:bridge-rep-multi} \\
&=& \Eof{\left. \prod_{k=1}^n \frac1{\sqrt{2\pi y_0^2\bar\rho^2 \Delta v_{t_k}}} e^{-\frac1{2 y_0^2 \bar\rho^2 \Delta v_{t_k}}{\left(\Delta x_{t_k} - y_0 \rho \int_{t_{k-1}}^{t_k} e^{\nu B^H_s} dB_s + \frac{y_0^2}2 \Delta v_{t_k} \right)^2}} \right| \nu \bB^H_{\bt} = \bseta_{\bt}} \times \nonumber \\
&& \P\left[y_0 e^{\nu \bB^H_{\bt}} = \by_{\bt} \right], \nonumber
\eea
where $\bseta_{\bt} = \log \by_{\bt} - \log y_0$, $\Delta x_{t_k} = x_{t_k} - x_{t_{k-1}}$, and $\Delta v_{t_k} = v_{t_k} - v_{t_{k-1}}$ for $k = 1, \cdots, n$. Recall that $v_t = \int_0^t e^{\nu B^H_s} ds$.
\end{theorem}
\begin{proof}
For any bounded measurable function $f:\R^n\times\R^n \to \R$, consider the expectation
\beaa
&& \iint f(\bx_{\bt}, \by_{\bt}) p(\bx_{\bt}, \by_{\bt}) d\bx_{\bt}d\by_{\bt} \\
&=& \Eof{f(\bX_{\bt}, \bY_{\bt})} \\
&=& \Eof{\Eof{\left. f(\bX_{\bt}, \bY_{\bt})\right|\cF^B_{t_n}}}.
\eeaa
Let $\xi_{t_i} = \int_0^{t_i} e^{\nu B_s^H} dW_s$, $\zeta_{t_i} = \int_0^{t_i} e^{\nu B_s^H} dB_s$ and thus accordingly $\Delta \xi_{t_i} = \xi_{t_i} - \xi_{t_{i-1}} = \int_{t_{i-1}}^{t_i} e^{\nu B_s^H} dW_s$, $\Delta \zeta_{t_i} = \zeta_{t_i} - \zeta_{t_{i-1}} = \int_{t_{i-1}}^{t_i} e^{\nu B_s^H} dB_s$. Note that, conditioned on $\cF_{t_n}^B$, the random variables $\Delta\xi_{t_i}$'s are independent normal with mean $0$ and variance $\Delta v_{t_i}$.
We calculate the conditional expectation as follows.
\bea
&& \Eof{\left. f(\bX_{\bt}, \bY_{\bt})\right|\cF^B_{t_n}} \nonumber \\
&=& \Eof{\left. f\left(x_0 + \rho y_0 \bszeta_{\bt} - \frac{y_0^2}2 \bv_{\bt} + \bar\rho y_0 \bsxi_{\bt}, y_0 e^{\nu \bB_{\bt}^H} \right)\right|\cF^B_{t_n}} \nonumber \\
&=& \int f\left(x_0 + \rho y_0 \bszeta_{\bt} - \frac{y_0^2}2 \bv_{\bt} + \bar\rho y_0 \bsxi_{\bt}, y_0 e^{\nu \bB_{\bt}^H} \right) \prod_{k=1}^n \frac1{\sqrt{2\pi\Delta v_{t_k}}} e^{-\frac{(\Delta\xi_{t_k})^2}{2 \Delta v_{t_k}}} d\Delta \bsxi_{\bt}. \label{eqn:f-int-multip}
\eea
By applying the change of variables
\beaa
&& x_{t_k} = x_0 + \rho y_0 \zeta_{t_k} - \frac{y_0^2}2 v_{t_k} + \bar\rho y_0 \xi_{t_k} \, ,
\eeaa
thus
\beaa
&& \Delta \xi_{t_k} = \frac1{\bar\rho y_0} \left(\Delta x_{t_k} - \Delta \zeta_{t_k} - \frac{y_0^2}2 \Delta v_{t_k}\right).
\eeaa
The integral \eqref{eqn:f-int-multip} becomes
\beaa
&& \int f\left( \bx_{\bt}, y_0 e^{\nu \bB_{\bt}^H} \right) \prod_{k=1}^n \frac1{\sqrt{2\pi\Delta v_{t_k}}} e^{-\frac{(\Delta\xi_{t_k})^2}{2\Delta v_{t_k}}} d \Delta \bsxi_{\bt} \\
&=& \int f\left( \bx_{\bt}, y_0 e^{\nu \bB_{\bt}^H} \right) \prod_{k=1}^n \frac1{\sqrt{2\pi\bar\rho^2 y_0^2 \Delta v_{t_k}}} e^{-\frac1{2\bar\rho y_0\Delta v_{t_k}}\left( \Delta x_{t_k} - \rho y_0 \int_{t_{k-1}}^{t_k} e^{\nu B_s^H} dB_s - \frac{y_0^2}2 \Delta v_{t_k} \right)^2} d \Delta \bx_{\bt} \\
&=& \int f\left( \bx_{\bt}, y_0 e^{\nu \bB_{\bt}^H} \right) \prod_{k=1}^n \frac1{\sqrt{2\pi\bar\rho^2 y_0^2 \Delta v_{t_k}}} e^{-\frac1{2\bar\rho y_0\Delta v_{t_k}}\left( \Delta x_{t_k} - \rho y_0 \int_{t_{k-1}}^{t_k} e^{\nu B_s^H} dB_s - \frac{y_0^2}2 \Delta v_{t_k} \right)^2} d\bx_{\bt}
\eeaa
since the Jacobian between $d\Delta \bx_{\bt}$ and $d\bx_{\bt}$ 1. It follows that
\beaa
&& \iint f(\bx_{\bt}, \by_{\bt}) p(\bx_{\bt}, \by_{\bt}) d\bx_{\bt}d\by_{\bt} \\
&=& \Eof{\left. \Eof{f(\bX_{\bt}, \bY_{\bt}) \right| \cF_t^B}} \\
&=& \Eof{\int f\left( \bx_{\bt}, y_0 e^{\nu \bB_{\bt}^H} \right) \prod_{k=1}^n \frac1{\sqrt{2\pi\bar\rho^2 y_0^2 \Delta v_{t_k}}} e^{-\frac1{2\bar\rho y_0\Delta v_{t_k}}\left( \Delta x_{t_k} - \rho y_0 \int_{t_{k-1}}^{t_k} e^{\nu B_s^H} dB_s - \frac{y_0^2}2 \Delta v_{t_k} \right)^2} d\bx_{\bt}} \\
&=& \iint d\bx_{\bt} d\by_{\bt} f\left( \bx_{\bt}, \by_{\bt} \right) \times \\
&& \Eof{\left. \prod_{k=1}^n \frac1{\sqrt{2\pi\bar\rho^2 y_0^2 \Delta v_{t_k}}} e^{-\frac1{2\bar\rho y_0\Delta v_{t_k}}\left( \Delta x_{t_k} - \rho y_0 \int_{t_{k-1}}^{t_k} e^{\nu B_s^H} dB_s - \frac{y_0^2}2 \Delta v_{t_k} \right)^2} \right|\nu \bB^H_{\bt} = \bseta_{\bt}} \times \\
&& \P\left[y_0 e^{\nu \bB^H_{\bt}} = \by_{\bt} \right].
\eeaa
This completes the proof of bridge representation \eqref{eqn:bridge-rep-multi} since $f$ is arbitrary.
\end{proof}

To move onto a heuristic derivation of the sample path large deviation principle for $(X_t, Y_t)$ in small time, we take logarithm on both sides of \eqref{eqn:bridge-rep-multi} and obtain
\bea
&& \log p(x_{t_1}, y_{t_1}, \cdots, x_{t_n}, y_{t_n}) \nonumber \\
&=& \log \Eof{\left. \prod_{k=1}^n \frac1{\sqrt{2\pi y_0^2\bar\rho^2 \Delta v_{t_k}}} e^{-\frac1{2 y_0^2 \bar\rho^2 \Delta v_{t_k}}{\left(\Delta x_{t_k} - y_0 \rho \int_{t_{k-1}}^{t_k} e^{\nu B^H_s} dB_s + \frac{y_0^2}2 \Delta v_{t_k} \right)^2}} \right| \nu \bB^H_{\bt} = \bseta} \nonumber \\
&& + \log \P\left[\nu \bB^H_{\bt} = \bseta_{\bt} \right] - \sum \log y_{t_i}. \label{eqn:log-p-multip}
\eea
In the following, we ignore the last term on the right hand side of \eqref{eqn:log-p-multip} and intuitively calculate the limits as $n\to\infty$ of the first two terms. Note that to the leading order we have
\beaa
&& \log \P\left[\nu \bB^H_{\bt} = \bseta_{\bt} \right] \approx -\frac1{2\nu^2} \boldsymbol{\eta}' \bsR^{-1} \boldsymbol{\eta},
\eeaa
where $\bsR = [R(t_i, t_j)]$ is the covariance matrix of $\bB^H_{\bt}$. We further discretize the autovariance $R$ of fractional Brownian motion as
\beaa
&& R(t_i, t_j) = \Eof{B^H_{t_i}B^H_{t_i}} = \int_0^{t_i\wedge t_j} K_H(t_i, s) K_H(t_j, s) ds \\
&\approx& \sum_{k=0}^{i\wedge j} K_H(t_i, t_k) K_H(t_j,t_k) \Delta t = \boldsymbol{K}' \boldsymbol{K} \Delta t,
\eeaa
where $\boldsymbol K$ denotes the upper triangular matrix
\[
\boldsymbol K_{ij} =
\left\{\begin{array}{ll}
K_H(t_i, t_j), & \mbox{ if } i \geq j; \\
& \\
0, & \mbox{ otherwise.}
\end{array}\right.
\]
Thereby, $\bsR^{-1} = \frac1{\Delta t} \boldsymbol K^{-1} (\boldsymbol K')^{-1}$. Let $\bsb = (b_{t_1}, \cdots b_{t_n})$ be the solution to the linear system
\[
\frac{\boldsymbol\eta}\nu = \boldsymbol K \bsb \Delta t.
\]
It follows that
\beaa
&& \frac1{2\nu^2} \boldsymbol{\eta}' \bsR^{-1} \boldsymbol{\eta}
= \frac12 \Delta t \bsb' \boldsymbol{K}' \bsR^{-1} \boldsymbol{K} \bsb \Delta t \\
&=& \frac12 \bsb' \bsb \Delta t
= \frac12 \sum_{k=1}^n b_{t_k}^2 \Delta t \\
&\longrightarrow& \frac12 \int_0^T b_t^2 dt \quad \mbox{ as } n \to \infty.
\eeaa
Also in the limit as $n \to \infty$, we obtain $\eta_t = \nu \int_0^t K_H(t,s) b_s ds$.

On the other hand, for the first term on the right hand side of \eqref{eqn:log-p-multip}, we have
\beaa
&& \log \Eof{\left. \prod_{k=1}^n \frac1{\sqrt{2\pi y_0^2\bar\rho^2 \Delta v_{t_k}}} e^{-\frac1{2 y_0^2 \bar\rho^2 \Delta v_{t_k}}{\left(\Delta x_{t_k} - y_0 \rho \int_{t_{k-1}}^{t_k} e^{\nu B^H_s} dB_s + \frac{y_0^2}2 \Delta v_{t_k} \right)^2}} \right| \nu \bsB^H = \bseta} \\
&\approx& \sum_{k=1}^n \Eof{\left. -\frac1{2 y_0^2 \bar\rho^2 \Delta v_{t_k}}{\left(\Delta x_{t_k} - y_0 \rho \int_{t_{k-1}}^{t_k} e^{\nu B^H_s} dB_s \right)^2} \right|\nu \bsB^H = \bseta}.
\eeaa
Note that conditioned on $\nu \bsB^H = \bseta$, we have
\beaa
&& \Delta v_{t_k} = \int_{t_{k-1}}^{t_k} e^{2 \nu B^H_s} ds \approx e^{2 \eta_{t_{k-1}}} \Delta t = e^{2\nu\sum_{j=0}^{k-1} K_H(t_{k-1}, t_j) b_{t_j} \Delta t} \Delta t
\eeaa
as well as
\beaa
&& \Delta x_{t_k} - y_0 \rho \int_{t_{k-1}}^{t_k} e^{\nu B^H_s} dB_s \approx
\Delta x_{t_k} - y_0 \rho e^{\eta_{t_{k-1}}} b_{t_{k-1}} \Delta t \\
&=& \left(\frac{\Delta x_{t_k}}{\Delta t} - y_0 \rho e^{\nu\sum_{j=0}^{k-1} K_H(t_{k-1}, t_j) b_{t_j} \Delta t} b_{t_{k-1}} \right) \Delta t.
\eeaa
It follows that the first term in \eqref{eqn:log-p-multip} has the limit
\beaa
&& \sum_{k=1}^n \Eof{\left. -\frac1{2 y_0^2 \bar\rho^2 \Delta v_{t_k}}{\left(\Delta x_{t_k} - y_0 \rho \int_{t_{k-1}}^{t_k} e^{\nu B^H_s} dB_s \right)^2} \right|\nu \bsB^H = \bseta} \\
&\approx& -\sum_{k=0}^n \frac1{2y_0^2 \bar\rho^2 e^{2\nu\sum_{j=0}^{k-1} K_H(t_{k-1}, t_j) b_{t_j} \Delta t}} \left(\frac{\Delta x_{t_k}}{\Delta t} - y_0 \rho e^{\nu\sum_{j=0}^{k-1} K_H(t_{k-1}, t_j) b_{t_j} \Delta t} b_{t_{k-1}} \right)^2 \Delta t \\
&\longrightarrow& -\frac12 \int_0^T \frac1{y_0^2 \bar\rho^2 e^{2 \nu \int_0^t K_H(t,s)b_s ds}} \left(\dot x_t - y_0 \rho e^{\nu \int_0^t K_H(t,s) b_s ds} b_t\right)^2 dt
\eeaa
as $n \to \infty$.

Putting the two limits together, we obtain heuristically for $T \approx 0$ that
\bea
&& -\log \Pof{X_t = x_t, Y_t = y_t, \mbox{ for } t \in [0,T]} \nonumber \\
&\approx& \frac12 \int_0^T \frac1{y_0^2 \bar\rho^2 e^{2\nu \int_0^t K_H(t,s)b_s ds}} \left(\dot x_t - y_0 \rho e^{\nu\int_0^t K_H(t,s) b_s ds} b_t\right)^2 dt + \frac12 \int_0^T b_t^2 dt \nonumber \\
&=& \frac12 \int_0^T \frac1{\bar\rho^2 y_t^2} \left(\dot x_t - \rho y_t b_t\right)^2 dt + \frac12 \int_0^T b_t^2 dt  \nonumber \\
&=& \frac12 \int_0^T \frac1{\bar\rho^2}\left(\frac{\dot x_t}{y_t} - \rho b_t\right)^2 dt + \frac12 \int_0^T b_t^2 dt, \label{eqn:energy-funcl}
\eea
where $b \in L^2[0,T]$ satisfies the integral equation
\[
\log y_t - \log y_0 = \nu \int_0^t K_H(t,s) b_s ds
\]
for all $t \in [0,T]$. We remark that \eqref{eqn:energy-funcl} should serve as the rate function for the sample path large deviation principle in small time for $(X_t,Y_t)$. Moreover, one may define the ``geodesic" from the initial point $(x_0, y_0)$ to the terminal point $(x_T, y_T)$ in the fSABR plane as the path $(x_t^*, y_t^*)$ which minimizes the functional \eqref{eqn:energy-funcl}, i.e.,
\beaa
(x_t^*, y_t^*) := \mathop{\argmin}\limits_{t\mapsto (x_t,y_t)} \frac12 \int_0^T \frac1{\bar\rho^2}\left(\frac{\dot x_t}{y_t} - \rho b_t\right)^2 dt + \frac12 \int_0^T b_t^2 dt,
\eeaa
where again $b_t$ is determined by solving the integral equation
\begin{equation}
\log y_t - \log y_0 = \nu \int_0^t K_H(t,s) b_s ds. \label{eqn:y-b-int-eq}
\end{equation}
Also, the minimizer can be regarded as the ``geodesic" connecting $(x_0, y_0)$ and $(x_T, y_T)$. 
\begin{remark}
Note that $b_t$ is indeed determined by the inverse operator $K_H^{-1}$ applied to $\log\frac{y_t}{y_0}$. In particular, with $H = \frac12$ this inverse operator reduces to the usual derivative. Thus, with $H = \frac12$,
\[
b_t = \frac{d}{dt}\left(\log\frac{y_t}{y_0}\right) = \frac{\dot y_t}{y_t}.
\]
The functional \eqref{eqn:energy-funcl} becomes
\beaa
&& \frac12 \int_0^T \frac1{\bar\rho^2}\left(\frac{\dot x_t}{y_t} - \rho b_t\right)^2 dt + \frac12 \int_0^T b_t^2 dt \\
&=& \frac12 \int_0^T \frac1{\bar\rho^2}\left(\frac{\dot x_t}{y_t} - \rho \frac{\dot y_t}{y_t} \right)^2 dt + \frac12 \int_0^T \frac{\dot y_t^2}{y_t^2} dt \\
&=& \frac12 \int_0^T \frac1{\bar\rho^2 y_t^2}\left(\dot x_t^2 - 2 \rho \dot x_t \dot y_t + \dot y_t^2\right) dt.
\eeaa
The last expression is the energy functional (up to the constant factor $\frac12$) associated with the Riemann metric $ds^2 = \frac1{\bar\rho^2 y^2}(dx^2 - 2\rho dx dy + dy^2)$. The diffusion process associated with this Riemann metric is governed by the SDEs
\beaa
&& dX_t = Y_t dW_t, \\
&& dY_t = Y_t dZ_t,
\eeaa
where $W_t$ and $Z_t$ are correlated Brownian motion with constant correlation $\rho$, which up to a linear transformation is the upper plane model of the Poincar\'e space. In other words, with $H = \frac12$, the functional \eqref{eqn:energy-funcl} recovers the energy functional for the classical Poincar\'e space, which is isometric to the SABR plane.
\end{remark}
Lastly, with the aid of sample path large deviation principle \eqref{eqn:energy-funcl}, it is nearly a common practice, say by applying the Laplace asymptotic formula, to conclude that the log premium of an out-of-money call in small time has the asymptotic as $t \to 0$
\[
-\log C(k, t) \approx -\log \Pof{X_t \geq k} \approx \frac12 \int_0^T \frac1{\bar\rho^2}\left(\frac{\dot x_t^*}{y_t^*} - \rho b_t^*\right)^2 dt + \frac12 \int_0^T {b_t^*}^2 dt,
\]
where $(x_t^*, y_t^*, b_t^*)$ denotes the optimal path that minimizes the functional \eqref{eqn:energy-funcl} subject to the constraint $x_t^* = k$ and $y_t^*$, $b_t^*$ satisfy the integral equation \eqref{eqn:y-b-int-eq}. Thus, by applying \eqref{eqn:iv-small-time-asymp}, an approximation of implied volatility in small time is readily obtained. We summarize the result in the following proposition which, with $H = \frac12$, recovers the SABR formula \eqref{eqn:sabr-formula}. However, for $H \neq \frac12$, the numerical implementation of \eqref{eqn:fSABR} is more involved than that of \eqref{eqn:fSABR-fake} since, as opposed to a one dimensional optimization problem, it is subject to solving a two-dimensional constrained variational problem.
\begin{proposition} (fSABR formula) \label{thm:fSABR} \\
Let $k = \log\left(\frac K{s_0}\right)$ be the log moneyness.
The implied volatility $\sigma_{\BS}(k,t)$ in small time to expiry has the asymptotics
\begin{equation}
\sigma_{\BS}^2 \approx \frac{k^2}T \left( \int_0^T \left\{\frac1{\bar\rho^2 {y_t^*}^2} \left(\dot x_t^* - \rho y_t^* b_t^* \right)^2 + {b_t^*}^2 \right\} dt \right)^{-1}, \label{eqn:fSABR}
\end{equation}
where $(x_t^*,b^*)$ is the minimizer of the variational problem
\[
(x^*,b^*) = \mbox{argmin} \left\{ \dot x, b \in L^2[0,T]:\int_0^T \left(\frac1{\bar\rho^2 y_t^2} \left(\dot x_t - \rho y_t b_t\right)^2 + b_t^2\right) dt \right\}
\]
with $x_T = k$ and $y_t^*$ satisfying
\[
\log y^*_t - \log y_0 = \nu \int_0^t K_H(t,s) b^*_s ds
\]
for $t \in [0,T]$. Notice that \eqref{eqn:fSABR} recovers the SABR formula \eqref{eqn:sabr-formula} with $H = \frac12$.
\end{proposition}

%
%

\section{Conclusion and discussion} \label{sec:conclusion}
We showed in this paper a bridge representation in Fourier space and a small time asymptotic for the joint probability of lognormal fractional SABR model for general $\rho \in (-1,1)$.
An application of the asymptotics of the joint density is an approximation of the implied volatility in small time. Due to the different nature of methodologies, the newly obtained approximation of implied volatilities in small time does not recover the celebrated SABR formula for implied volatility (to the zeroth order) when the Hurst exponent $H$ equals a half. To recover the SABR formula, we presented a heuristic derivation of the sample path large deviation principle for the lognormal fractional SABR model by bootstrapping via the multiperiod joint density. We emphasize once again that the same trick is applicable to general fractional SABR models, i.e., to include a local volatility component in the process $S_t$ for underlying asset. We leave the rigorous proof of the sample path large deviation principle for fractional SABR models in a future work. Lastly, the bridge representation methodology is also applicable to the case in which the volatility process is governed by an exponential fractional Ornstein-Uhlenbeck process since a fractional Ornstein-Uhlenbeck process is Gaussian as well. However, as time to expiry approaches zero, the mean reversion part does not really play a role in the large deviation regime.


\section*{Acknowledgement} 
We are grateful for helpful discussions with the participants of the conferences: At the Frontiers of Quantitative Finance at
International Centre of Mathematical Sciences, Edinburgh, UK and Mathematics of Quantitative Finance at Mathematisches Forschungsinstitut Oberwolfach, Oberwolfach, Germany. JA is supported by JSPS KAKENHI Grant Number $23330109$,
$24340022$, $23654056$ and $25285102$, and the project RARE-318984 (an FP7 Marie Curie IRSES).
THW is partially supported by the Natural Science Foundation of China grant 11601018.

%
%

\section{Appendix - Technical proofs}
In the appendix, we provide a detailed error analysis of the asymptotic expansion for \eqref{eqn:small-time-p} and a version of Laplace asymptotic formula that is readily applicable to our case. 

\subsection{Error analysis} \label{sec:error-estimate}
Let $C^\infty_0(\R^2)$ be the space of smooth functions defined on $\R^2$ with compact support.
For a given $f \in C^\infty_0(\R^2)$, recalling $\eta=\ln{(y/y_0)}$, from \eqref{eqn:bridge-rep} we have
\bea
&& \Eof{f(X_t, Y_t)} = \iint f(x, y) p(t;x, y) dx\, dy \nonumber \\
&=& \frac1{2\pi}\iiint f(x, y) \frac{e^{-\frac{\eta^2}{2\nu^2 t^{2H}}}}{ y\sqrt{2\pi \nu^2 t^{2H}}} e^{i(x - x_0)\xi} \E_{\frac\eta\nu} \left[ e^{i\left(-\rho \int_0^t y_0 e^{\nu B^H_s} dB_s + \frac{y_0^2 v_t}2 \right)\xi} e^{-\frac{\bar\rho^2 y_0^2 v_t\xi^2}2 }\right] d\xi dx dy \nonumber \\
&=& \frac1{2\pi}\iint e^{-ix_0\xi} \hat f(\xi, y) \frac{e^{-\frac{\eta^2}{2\nu^2 t^{2H}}}}{ y\sqrt{2\pi \nu^2 t^{2H}}} \E_{\frac\eta\nu} \left[ e^{i\left(-\rho \int_0^t y_0 e^{\nu B^H_s} dB_s + \frac{y_0^2v_t}2  \right)\xi} e^{-\frac{\bar\rho^2 y_0^2 v_t\xi^2}2 }\right] dyd\xi \nonumber \\
&=& \frac1{2\pi}\int e^{-ix_0\xi} \Eof{\hat f(\xi, Y_t) e^{i\left(-\rho \int_0^t y_0 e^{\nu B^H_s} dB_s + \frac{y_0^2v_t}2  \right)\xi} e^{-\frac{\bar\rho^2 y_0^2  v_t\xi^2}2}} d\xi, \label{eqn:Ef}
\eea
where
\[
\hat f(\xi,y) = \int e^{i\xi x} f(x, y) dx
\]
is the Fourier transform of $f$ with respect to $x$.

Note that the right-hand side of \eqref{eqn:small-time-p} equals the right-hand side of \eqref{eqn:bridge-rep-0th}. We compare \eqref{eqn:Ef} with the following expression obtained by using the approximate joint density in \eqref{eqn:small-time-p} and obtain
\bea\label{eqn:Ef-approx}
&& \frac1{2\pi} \iiint f(x, y) \frac{e^{-\frac{\eta^2}{2\nu^2t^{2H}}}}{ y\sqrt{2\pi \nu^2 t^{2H}}} \nonumber \\
&& \quad\times e^{i\left(x- x_0\right)\xi}e^{-\frac12(\bar{\rho}^2\xi-i)\xi\int_0^ty_0^2e^{2\nu m_s}ds}\mathbb{E}_{\frac\eta\nu}\left[e^{-i\rho\xi\int_0^ty_0e^{\nu m_s}dB_s}\right]
 \, d\xi dx dy\nonumber\\
 &=& \frac1{2\pi} \iint e^{-ix_0\xi}\hat{f}(\xi,y)\frac{e^{-\frac{\eta^2}{2\nu^2t^{2H}}}}{ y\sqrt{2\pi \nu^2 t^{2H}}} \nonumber \\
&& \quad\times\mathbb{E}_{\frac\eta\nu}\left[e^{i\left(-\rho\int_0^ty_0e^{\nu m_s}dB_s+\frac{y_0^2}{2}\int_0^te^{2\nu m_s}ds\right)\xi}e^{-\frac{\bar{\rho}^2y_0^2\xi^2}{2}\int_0^te^{2\nu m_s}ds}\right]
 \, d\xi dx dy \nonumber\\
&=& \frac1{2\pi} \int e^{-ix_0\xi} \Eof{\hat f(\xi, Y_t) e^{i\left(-\rho\int_0^ty_0e^{\nu m_s}dB_s+\frac{y_0^2}{2}\int_0^te^{2\nu m_s}ds\right)\xi}e^{-\frac{\bar{\rho}^2y_0^2\xi^2}{2}\int_0^te^{2\nu m_s}ds}} d\xi. \nonumber\\
\eea
For simplification, denote
\[
\lambda_1(t)=e^{i\left(-\rho \int_0^t y_0 e^{\nu B^H_s} dB_s + \frac{y_0^2v_t}2  \right)\xi} e^{-\frac{\bar\rho^2 y_0^2 v_t}2 \xi^2}
\]
and 
\[
\lambda_2(t)=e^{i\left(-\rho\int_0^ty_0e^{\nu m_s}dB_s+\frac{y_0^2}{2}\int_0^te^{2\nu m_s}ds\right)\xi}e^{-\frac{\bar{\rho}^2y_0^2\xi^2}{2}\int_0^te^{2\nu m_s}ds}.
\]
Then the modulus of the difference between \eqref{eqn:Ef} and \eqref{eqn:Ef-approx} is equal to
\bea\label{eqn-diff}
 \left|\frac1{2\pi}\int e^{-ix_0\xi} \Eof{\hat f(\xi, Y_t) \left(\lambda_1(t)-\lambda_2(t)\right)} d\xi \right|. 
\eea
The goal is to show that \eqref{eqn-diff} converges to zero in the order of $t^{\frac12}$ as $t \to 0$, for every $f \in C^\infty_0(\R^2)$.

By applying the following inequality, for any $z, w \in \mathbb{C}$,
\[
\left| e^z - e^w \right| \leq \left(e^{\Re(z)} + e^{\Re(w)}\right) |z - w|,
\]
where $\Re(z)$ denotes the real part of $z$, we have
\bea\label{eqn:RI}
&&|\lambda_1(t)-\lambda_2(t)|\nonumber\\
&\leq& \left(e^{-\frac{\bar\rho^2 y_0^2 v_t\xi^2}2 } + e^{-\frac{\bar\rho^2 y_0^2 \xi^2}2\int_0^te^{2\nu m_s}ds} \right) \times \nonumber\\
&& \; \left| i\left(-\rho \int_0^t y_0 e^{\nu B^H_s} dB_s + \frac{y_0^2v_t }2 \right)\xi - \frac{\bar\rho^2 y_0^2 v_t\xi^2 }2 \right.  \nonumber\\
&& \left. - i\left(-\rho\int_0^ty_0e^{\nu m_s}dB_s+\frac{y_0^2}{2}\int_0^te^{2\nu m_s}ds\right)\xi+\frac{\bar{\rho}^2y_0^2\xi^2}{2}\int_0^te^{2\nu m_s}ds\right|  \nonumber\\
&\leq& 2\left| \cR_t + i \cI_t \right|
\eea
since $e^{-\frac{\bar\rho^2 y_0^2 v_t}2 \xi^2} + e^{-\frac{\bar\rho^2 y_0^2 \xi^2}2\int_0^te^{2\nu m_s}ds} \leq 2$ for all $t$ and $\xi$.
Apparently, $\cR_t$ and $\cI_t$ are given by
\beaa
&& \cR_t = \left[- v_t +\int_0^te^{2\nu m_s}ds \right] \frac{\bar{\rho}^2y_0^2\xi^2}2 , \\
&& \cI_t = \left(-\rho \int_0^t y_0 e^{\nu B^H_s} dB_s + \frac{y_0^2v_t }2 + \rho\int_0^ty_0e^{\nu m_s}dB_s-\frac{y_0^2}{2}\int_0^te^{2\nu m_s}ds\right)\xi.
\eeaa

In the following, $K$ denotes a generic constant whose value may vary in different contexts. Then, by \eqref{eqn-diff}, \eqref{eqn:RI} and 
H\"older's inequality, we have
\bea\label{eqn:hatf-RI}
&&\left|\frac1{2\pi}\int e^{-ix_0\xi} \Eof{\hat f(\xi, Y_t) \left(\lambda_1(t)-\lambda_2(t)\right)} d\xi \right|\nonumber \\
&\leq& 2 \int \Eof{|\hat f(\xi,Y_t)| \left| \cR_t + i \cI_t \right|} d\xi\nonumber\\
&\leq& 2\left(\E\int |\hat f(\xi, Y_t)|^{(1-\epsilon) p} d\xi \right)^{\frac1p} \;
\left(\int \Eof{|\hat f(\xi, Y_t)|^{\epsilon q} \left| \cR_t + i \cI_t \right|^q} d\xi \right)^{\frac1q} \nonumber \\
&\leq& K \left(\E\int |\hat f(\xi, Y_t)|^{(1-\epsilon) p} d\xi \right)^{\frac1p} \;
\left(\int \Eof{|\hat f(\xi, Y_t)|^{\epsilon q} (|\cR_t|^q + |\cI_t|^q)} d\xi \right)^{\frac1q},
\eea
for some $\epsilon \in (0,1)$ and $\frac1p + \frac1q = 1$, $p,q > 0$.

Since $ f \in C^\infty_0(\R^2)$, it is easy to show the following properties of $\hat f$:
\begin{itemize}
\item[(i)] for any $r\geq 0$, $\displaystyle{\sup_{(\xi,y)\in \mathbb{R}^2}\left|\xi^r\hat f(\xi,y)\right|<\infty}$;
\item[(ii)] for any $r\geq 0$ and  $p>0$, $\displaystyle{\int |\xi|^r \sup_{y\in\mathbb{R}}|\hat f(\xi,y)|^pd\xi}<\infty$.
\end{itemize}
Note that property (ii) can be easily obtained by property (i).

By the above property (ii),  we can show that
\begin{equation}\label{eqn:hatf-int}
\limsup_{t\to0^+} \E\int |\hat f(\xi, Y_t)|^{(1-\epsilon) p} d\xi < \infty.
\end{equation}
We compute the second term in \eqref{eqn:hatf-RI} separately as follows. By changing variables, we get
\bea\label{eqn:L12}
&& \int \Eof{|\hat f(\xi,Y_t)|^{\epsilon q} |\cR_t|^q} d\xi\nonumber \\
&\leq& K\bar{\rho}^{2q} y_0^{2q} \int \Eof{|\hat f(\xi, Y_t)|^{\epsilon q} \left(v_t^q +\left(\int_0^te^{2\nu m_s}ds\right)^q \right)} \xi^{2q} d\xi \nonumber\\
&=& K\bar{\rho}^{2q} y_0^{2q}  t^q  \int \Eof{|\hat f(\xi, Y_t)|^{\epsilon q}\left( \left(\int_0^1 e^{2\nu B^H_{tu}} du\right)^q + \left(\int_0^te^{2 R(1,u)\eta}du\right)^q \right)} \xi^{2q} d\xi \nonumber\\
&=& K\bar{\rho}^{2q} y_0^{2q}  t^q ( L_1 + L_2),
\eea
where
\beaa
&& L_1 := \int \xi^{2q} \Eof{|\hat f(\xi, Y_t)|^{\epsilon q} \left(\int_0^1 e^{2\nu B^H_{tu}} du \right)^q} d\xi,  \\
&& L_2 := \int \xi^{2q} \Eof{|\hat f(\xi, Y_t)|^{\epsilon q} \left(\int_0^1 e^{2 R(1,u)\eta} du \right)^q} d\xi.
\eeaa
By property (ii) of $\hat f$, it is easy to see that
\bea\label{eqn:L2}
 \limsup_{t\to 0^+} L_2 \leq \left(\int_0^1 e^{2 R(1,u)\eta} du \right)^q\int \xi^{2q} \Eof{|\hat f(\xi, Y_t)|^{\epsilon q}}d\xi<\infty.
\eea
For $L_1$, by Jensen's inequality and H\"{o}lder's inequality, we have
\beaa
L_1&\leq& \int \xi^{2q} \Eof{|\hat f(\xi, Y_t)|^{\epsilon q} \int_0^1 e^{2q\nu B^H_{tu}} du} d\xi\\
&\leq& \int \xi^{2q} \left\{\Eof{|\hat f(\xi, Y_t)|^{\epsilon q p_1}}\right\}^{\frac1{p_1}} \left\{\Eof{\left(\int_0^1 e^{2q\nu B^H_{tu}} du \right)^{q_1}}\right\}^{\frac1{q_1}}d\xi \\
&\leq&  \int \xi^{2q} \left\{\Eof{|\hat f(\xi, Y_t)|^{\epsilon q p_1}}\right\}^{\frac1{p_1}} \left\{\int_0^1 \Eof{e^{2q q_1\nu B^H_{tu}}} du\right\}^{\frac1{q_1}} d\xi\\
&=&  \int \xi^{2q}\left\{\Eof{|\hat f(\xi, Y_t)|^{\epsilon q p_1}}\right\}^{\frac1{p_1}} \left\{\int_0^1 e^{2(qq_1\nu)^2(tu)^{2H}} du\right\}^{\frac1{q_1}}d\xi.
\eeaa
where $\frac1{p_1}+\frac1{q_1}=1$ with $p_1, q_1>0$.
Therefore, using property (ii) again, we can easily show 
\begin{equation}\label{eqn:L1}
 \limsup_{t\to 0^+} L_1 \leq \limsup_{t\to 0^+} \int \xi^{2q} \left\{\Eof{|\hat f(\xi, Y_t)|^{\epsilon q p_1}}\right\}^{\frac1{p_1}} d\xi \;\left\{\int_0^1 e^{2(qq_1\nu)^2(u)^{2H}} du\right\}^{\frac1{q_1}} < \infty.
\end{equation}
Thus, it implies from \eqref{eqn:L12}-\eqref{eqn:L1} that 
\begin{equation}\label{eqn:R}
\int \Eof{|\hat f(\xi,Y_t)|^{\epsilon q} |\cR_t|^q} d\xi = O(t^q),
\end{equation}
for any $q > 1$, as $t \to 0^+$.

Similarly, we can write
\beaa
&& \int \Eof{|\hat f(\xi,Y_t)|^{\epsilon q} |\cI_t|^q} d\xi \nonumber\\
&\leq& K \int |\xi|^q \E\left[|\hat f(\xi,Y_t)|^{\epsilon q} \times \right. \nonumber\\
&& \left. \left\{ \left|\rho \int_0^t y_0 e^{\nu B^H_s} dB_s\right|^q + \left|\frac{y_0^2v_t}2 \right|^q +  \left|\rho \int_0^t y_0 e^{\nu m_s} dB_s\right|^q  + \left|\frac{y_0^2}2 \int_0^te^{2\nu m_s}ds\right|^q \right\} \right] d\xi\nonumber\\
&=&K(J_1+J_2+J_3+J_4),
\eeaa
where
\beaa
&& J_1 := |\rho|^q \int |\xi|^q \Eof{|\hat f(\xi,Y_t)|^{\epsilon q} \left|\int_0^t y_0 e^{\nu B^H_s} dB_s\right|^q} d\xi, \\
&& J_2 := \int |\xi|^q \Eof{|\hat f(\xi,Y_t)|^{\epsilon q} \left|\frac{y_0^2v_t}2 \right|^q} d\xi, \\
&& J_3 := \int |\xi|^q \Eof{|\hat f(\xi,Y_t)|^{\epsilon q}  \left|\rho \int_0^t y_0 e^{\nu m_s} dB_s\right|^q} d\xi, \\
&& J_4 := \int |\xi|^q \Eof{|\hat f(\xi,Y_t)|^{\epsilon q} \left|\frac{y_0^2}2 \int_0^te^{2\nu m_s}ds\right|^q} d\xi.
\eeaa
We estimate $J_1$ through $J_4$ separately as follows.
\bit
\item $J_1$: Choosing $p_1>0$ such that $ \frac{qq_1}2 > 1$, by H\"{o}lder's inequality, the Burkholder-Davis-Gundy inequality, Jensen's inequality and a change of variables, we obtain
Notice that
\beaa\label{eqn:J1}
J_1 &\leq& |\rho|^q y_0^q\int |\xi|^q \left\{\Eof{|\hat f(\xi,Y_t)|^{\epsilon q p_1}}\right\}^{\frac1{p_1}} \left\{\Eof{\left|\int_0^t e^{\nu B^H_s} dB_s\right|^{qq_1}}\right\}^{\frac1{q_1}}d\xi\nonumber\\
&\leq& |\rho|^qy_0^q \int |\xi|^q\left\{\Eof{|\hat f(\xi,Y_t)|^{\epsilon q p_1}}\right\}^{\frac1{p_1}} \left\{\Eof{\left|\int_0^t e^{2\nu B^H_s} ds\right|^{\frac{qq_1}2}}\right\}^{\frac1{q_1}} d\xi \nonumber\\
&=& |\rho|^qy_0^q t^{\frac q2}\int |\xi|^q \left\{\Eof{|\hat f(\xi,Y_t)|^{\epsilon q p_1}}\right\}^{\frac1{p_1}} \left\{\Eof{\left|\int_0^1 e^{2\nu B^H_{tu}} du\right|^{\frac{qq_1}2}}\right\}^{\frac1{q_1}}d\xi \nonumber\\
&\leq&  |\rho|^qy_0^q t^{\frac q2} \int |\xi|^q\left\{\Eof{|\hat f(\xi,Y_t)|^{\epsilon q p_1}}\right\}^{\frac1{p_1}} \left\{\int_0^1 \Eof{e^{qq_1\nu B^H_{tu}}} du \right\}^{\frac1{q_1}}d\xi\nonumber \\
&=&  |\rho|^qy_0^q t^{\frac q2} \int|\xi|^q\left\{\Eof{|\hat f(\xi,Y_t)|^{\epsilon q p_1}}\right\}^{\frac1{p_1}} \left\{\int_0^1 e^{\frac{(qq_1\nu)^2}2 (tu)^{2H}} du \right\}^{\frac1{q_1}}d\xi.
\eeaa
By property (ii) we have
\beaa\label{eqn:J1-1}
 &&\limsup_{t\to 0^+}\int|\xi|^q\left\{\Eof{|\hat f(\xi,Y_t)|^{\epsilon q p_1}}\right\}^{\frac1{p_1}} \left\{\int_0^1 e^{\frac{(qq_1\nu)^2}2 (tu)^{2H}} du \right\}^{\frac1{q_1}}d\xi\nonumber\\
 &\leq & \limsup_{t\to 0^+} \int \xi^{2q} \left\{\Eof{|\hat f(\xi, Y_t)|^{\epsilon q p_1}}\right\}^{\frac1{p_1}} d\xi \;\left\{\int_0^1 e^{2(qq_1\nu)^2(u)^{2H}} du\right\}^{\frac1{q_1}} < \infty.
\eeaa
Thus, we can see that $J_1 = O(t^{\frac q2})$ as $t \to 0^+$.

\item $J_2$ and $J_4$: The asymptotic behavior of $J_2$ and $J_4$ is the same as that of $t^q L_1$, and hence, $J_2, J_4 = O(t^q)$ as $t \to 0^+$.

\item $J_3$: By using the same technique to $J_1$, we have
\[
J_3 \leq|\rho|^qy_0^q t^{\frac q2}\int |\xi|^q \left\{\Eof{|\hat f(\xi,Y_t)|^{\epsilon q p_1}}\right\}^{\frac1{p_1}} \left\{\Eof{\left|\int_0^1 e^{2 R(1,u)\eta} du\right|^{\frac{qq_1}2}}\right\}^{\frac1{q_1}}d\xi
\]
and $J_3 = O(t^{\frac q2})$ as $t \to 0^+$.
\eit
Thus, putting all the estimates for the $J_i$'s together we get
\begin{equation}\label{eqn:I}
\int \Eof{|\hat f(\xi,Y_t)|^{\epsilon q} |\cI_t|^q} d\xi = O(t^{\frac{q}2}),
\end{equation}
for any $q > 1$, as $t \to 0^+$.

Therefore, it implies from \eqref{eqn:hatf-RI}, \eqref{eqn:hatf-int}, \eqref{eqn:R} and \eqref{eqn:I} that

\beaa
 \left|\frac1{2\pi}\int e^{-ix_0\xi} \Eof{\hat f(\xi, Y_t) \left(\lambda_1(t)-\lambda_2(t)\right)} d\xi \right|=O(t^{\frac12}),
\eeaa
that is, \eqref{eqn-diff} converges to zero in the order of $t^{\frac12}$ as $t \to 0$, for every $f \in C^\infty_0(\R^2)$.

\subsection{Laplace asymptotic formula} \label{sec:laplace-formula}
We prove the following form of Laplace asymptotic formula required in the derivation of the small time asymptotic of the price of an out-of-money call.
\begin{lemma}(Laplace asymptotic formula) \label{lma:Laplace-asymp} \\
Let $\mathcal{C}$ be a closed and convex set in $\R^2$ with nonempty and smooth boundary $\p \mathcal{C}$. Suppose that $\theta(t, x) := \theta_0(x) + t^{\alpha} \theta_1(x) + t^{2\alpha} \theta_2(x)$, with $0\leq2\alpha < 1$, has   continuous second-order partial derivatives in $x\in \mathcal{C}$, and, for every $t$ sufficiently small, the function $\theta(t,x)$ is locally convex  in $\mathcal{C}$ and attains its minimum uniquely at $x^*(t) \in \p \mathcal{C}$. Moreover, there is $\epsilon_0>0$ such that for any $0<\epsilon <\epsilon_0$, there exist $t_0$ and $\delta > 0$ for which

 $$\theta(t, x) \geq \theta(t, x^*(t)) + \delta,\ \forall (t,x) \in [0,t_0]\times\left(\mathcal{C} \setminus B_\epsilon(x^*(t))\right),$$ 
  where $B_\epsilon(x^*(t)) = \{x:|x - x^*(t)| < \epsilon \}$ is the open ball of radius $\epsilon$ centered at $x^*(t)$. 
 
Assume that $f$ has  continuous second-order partial derivatives in $\mathcal{C}$, is integrable over $\mathcal{C}$ (i.e., $\int_\mathcal{C} |f(x)|dx < \infty$) and that $f$ vanishes identically in $\mathcal{C}^c$ and on the boundary $\p\mathcal{C}$ but the inward normal directional derivative of $f$ at $x^*(t)$ is nonzero.

Then, we have the asymptotic expansion, as $t \to 0^+$,
\begin{eqnarray} \label{eqn:lap-int-1}
&& \int_\mathcal{C} e^{-\frac{\theta(x,t)}t} f(x) dx\nonumber \\
&=& \frac{\sqrt{2\pi}\, t^{\frac52}\, e^{-\frac{\theta(t,x^*(t))}t}}{\sqrt{\p_{\mathbf{tan}}^2\theta(t,x^\ast(t))}|\nabla\theta(t,x^*(t))|}
\left[ \frac{\nabla f(x^*(t))\cdot\nabla\theta(t,x^*(t))}{|\nabla\theta(t,x^*(t))|^2} + \frac12 \frac{\p_{\mathbf{tan}}^2 f(x^*(t))}{\p_{\mathbf{tan}}^2 \theta(t,x^*(t))} +o(1)\right],\nonumber\\ 
\end{eqnarray}
where $\p_{\mathbf{tan}}^2 f(x^*)$ and $\p_{\mathbf{tan}}^2 \theta(t, x^*)$ are the second derivatives of $f$ and $\theta$ respectively in the tangential direction to $\mathcal{C}$ at $x^*$.
\end{lemma}
\begin{proof}
For any $0<\epsilon <\epsilon_0$, we split the integral on the left side of \eqref{eqn:lap-int-1} into two parts as
\begin{equation} \label{eqn:lap-int-2}
\int_\mathcal{C} e^{-\frac{\theta(t,x)}t} f(x) dx = \int_{\mathcal{C} \bigcap B_\epsilon(x^*(t))} e^{-\frac{\theta(t,x)}t} f(x) dx + \int_{\mathcal{C} \setminus B_\epsilon(x^*(t))} e^{-\frac{\theta(t,x)}t} f(x) dx.
\end{equation}
We treat the two terms on the right hand side of \eqref{eqn:lap-int-2} individually. For the first term, since the integration region is restricted to a subset of the small ball $B_\epsilon(x^*(t))$, it can be reparametrized by $y = (y^1, y^2)$ so that in the $y$-coordinates the set $\{ y: y^2 = 0 \}$ corresponds to $\p \mathcal{C}$ and the vectors $\{\p_{y^1}, \p_{y^2}\}$ form a local orthonormal frame around $x^*(t)$. For simplicity, we further assume that in the $y$-coordinates $x^*(t)$ is located at the origin. Note that in the $y$-coordinates the vector $\p_{y^2}$ is parallel to $\nabla \theta(x^*(t))$ as well as the inward normal vector of $\mathcal{C}$ at $x^*(t)$.

We shall use the convention that repeated indices are summed up over their respective ranges. Denote partial derivatives by subindices, we have for $y \in B_\epsilon(x^*(t))$
\beaa
  && \theta(t, y) = \theta(t, 0) + \theta_2(t, 0) y^2 + \frac12 \theta_{ij}(t, 0) y^i y^j + o(|y|^2),  \\
  && f(y) = f_i(0) y^i + \frac12 f_{ij}(0) y^i y^j + o(|y|^2)
\eeaa
since $\theta_1(0)= 0$ for $\theta$ attains its minimum at the boundary point $x^*(t)$. 

Thus, in the $y$-coordinates
the first integral on the right-hand side of \eqref{eqn:lap-int-2} reads
\bea\label{eqn:int-ball}
&& \int_{\mathcal{C} \bigcap B_\epsilon(x^*(t))} e^{-\frac{\theta(t,x)}t} f(x) dx \nonumber\\
&\approx& \int_0^\epsilon \int_{-\epsilon}^\epsilon e^{-\frac1t\left( \theta(t, 0) + \theta_2(t,0) y^2 + \frac12\theta_{ij}(t,0) y^i y^j  \right) } \left[ f_i(0) y^i + \frac12 f_{ij}(0) y^i y^j  \right]  dy^1 dy^2.
\eea
Now, by a change of variables
\[
y^1 = \sqrt t z^1, \quad y^2 = t z^2,
\]
we can write the above integral on the right-hand side of \eqref{eqn:int-ball} as
\bea\label{eqn:ytoz}
&& e^{-\frac{\theta(t,0)}t} t^{\frac32} \int_0^{\frac\epsilon t} \int_{-\frac{\epsilon}{\sqrt t}}^{\frac\epsilon{\sqrt t}} e^{-\left( \theta_2(t,0) z^2 + \frac12\theta_{11}(t,0) (z^1)^2 + \theta_{12}(t,0)z^1z^2\sqrt{t}+\frac12\theta_{22}(t,0)(z^2)^2t \right) } \times\nonumber \\
  && \left[ f_1(0) z^1 \sqrt t + f_2(0) z^2 t + \frac12 f_{11}(0) (z^1)^2 t +f_{12}(0)z^1z^2t^\frac32+\frac12 f_{22}(0)(z^2)^2t^2 \right]  dz^1 dz^2.\nonumber\\
   \eea
  Note that, for any real numbers $a_1,\dots, a_5$, by dominated convergence theorem, we have
  \beaa
 && \lim\limits_{t\to 0}\int_0^{\frac\epsilon t}\int_{-\frac{\epsilon}{\sqrt t}}^{\frac\epsilon{\sqrt t}} e^{-\left( \theta_2(t,0) z^2 + \frac12\theta_{11}(t,0) (z^1)^2 + \theta_{12}(t,0)z^1z^2\sqrt{t}+\frac12\theta_{22}(t,0)(z^2)^2t \right) } \times\nonumber \\
 && \left[a_1z^1+a_2z^2+a_3(z^1)^2+a_4(z^2)^2+a_5z^1z^2\right]dz^1dz^2\\
 &=&\int_0^{\infty}\int_{-\infty}^{\infty} e^{-\left( \theta_2(0,0) z^2 + \frac12\theta_{11}(0,0) (z^1)^2  \right) } \times\nonumber \\
 && \left[a_1z^1+a_2z^2+a_3(z^1)^2+a_4(z^2)^2+a_5z^1z^2\right]dz^1dz^2\in (-\infty, \infty).
  \eeaa
  
 Thus, the quantity in \eqref{eqn:ytoz} equals
  \bea
 && e^{-\frac{\theta(t,0)}t} t^{\frac32}\left\{ \int_0^{\frac\epsilon t} \int_{-\frac{\epsilon}{\sqrt t}}^{\frac\epsilon{\sqrt t}} e^{-\left( \theta_2(t,0) z^2 + \frac12\theta_{11}(t,0) (z^1)^2  \right) } \times\right. \nonumber\\
  && \left.\left[ f_1(0) z^1 \sqrt t + f_2(0) z^2 t + \frac12 f_{11}(0) (z^1)^2 t  \right]  dz^1 dz^2+O\left(t^{\frac12}\right)\right\}\nonumber\\
  &=&e^{-\frac{\theta(t,0)}t} t^{\frac32}\left[\sqrt{t}\cdot I+t\cdot II+t\cdot III+O\left(t^{\frac12}\right)\right],
\eea
where
\beaa
&& I = \int_0^{\frac\epsilon t} \int_{-\frac{\epsilon}{\sqrt t}}^{\frac\epsilon{\sqrt t}} e^{-\left( \theta_2(t,0) z^2 + \frac12\theta_{11}(t,0) (z^1)^2 \right) } f_1(0) z^1 dz^1 dz^2, \\
&& I\!I = \int_0^{\frac\epsilon t} \int_{-\frac{\epsilon}{\sqrt t}}^{\frac\epsilon{\sqrt t}} e^{-\left( \theta_2(t,0) z^2 + \frac12\theta_{11}(t,0) (z^1)^2  \right) } f_2(0) z^2 dz^1 dz^2, \\
&& I\!I\!I = \frac12 \int_0^{\frac\epsilon t} \int_{-\frac{\epsilon}{\sqrt t}}^{\frac\epsilon{\sqrt t}} e^{-\left( \theta_2(t,0) z^2 + \frac12\theta_{11}(t,0) (z^1)^2  \right) } f_{11}(0) (z^1)^2 dz^1 dz^2.
\eeaa
As $t \to 0^+$, we calculate the each integral individually as follows. For $I$, since the function in $z^1$ is an odd function and integral interval for $z^1$ is symmetric about the origin, we obtain
\bea
I &=& f_1(0) \int_0^{\frac\epsilon t}  e^{-\theta_2(t,0) z^2} dz^2 \times
\int_{-\frac{\epsilon}{\sqrt t}}^{\frac\epsilon{\sqrt t}} e^{-\left(\frac12\theta_{11}(t,0) (z^1)^2 \right)} z^1 dz^1\nonumber \\
&=& 0.
\eea

For $I\!I$ and $I\!I\!I$, notice that $\theta_2(t,0)> 0$ and $\theta_{11}(t,0)> 0$, and hence, we obtain
\bea
I\!I &=& f_2(0) \int_0^{\frac\epsilon t} e^{-\theta_2(t,0) z^2} z^2 dz^2 \times \int_{-\frac{\epsilon}{\sqrt t}}^{\frac\epsilon{\sqrt t}} e^{-\left(\frac12\theta_{11}(t,0) (z^1)^2  \right) } dz^1 \nonumber\\
&\approx& f_2(0) \int_0^\infty  e^{-\theta_2(t,0) z^2} z^2 dz^2 \times
\int_{-\infty}^\infty e^{-\frac12\theta_{11}(t,0) (z^1)^2} dz^1 \nonumber\\
&=& \frac{f_2(0)}{\theta_2^2(t,0)} \times \sqrt{\frac{2\pi}{\theta_{11}(t,0)}},
\eea
and
\bea\label{eqn:III}
I\!I\!I &=& \frac{f_{11}(0)}2 \int_0^{\frac\epsilon t} e^{- \theta_2(t,0) z^2} dz^2 \times \int_{-\frac{\epsilon}{\sqrt t}}^{\frac\epsilon{\sqrt t}} e^{-\left(\frac12\theta_{11}(t,0) (z^1)^2  \right) } (z^1)^2 dz^1\nonumber \\
&\approx& \frac{f_{11}(0)}2 \int_0^\infty e^{- \theta_2(t,0) z^2} dz^2 \times \int_{-\infty}^\infty e^{-\frac12\theta_{11}(t,0) (z^1)^2 } (z^1)^2 dz^1 \nonumber\\
&=& \frac{f_{11}(0)}{2\theta_2(t,0)} \times \sqrt{\frac{2\pi}{\theta_{11}^3(t,0)}}.
\eea

Therefore, it implies from \eqref{eqn:int-ball}-\eqref{eqn:III} that, in the $y$-coordinates,
\begin{equation}\label{eqn:first}
 \int_{\mathcal{C} \bigcap B_\epsilon(x^*(t))} e^{-\frac{\theta(t,x)}t} f(x) dx\approx e^{-\frac{\theta(t,0)}{t}}t^{\frac52}\sqrt{\frac{2\pi}{\theta_{11}(t,0)}}\left[\frac{f_2(0)}{\theta_2^2(t,0)}+ \frac{f_{11}(0)}{2\theta_2(t,0)\theta_{11}(t,0)}+o(1) \right].
\end{equation}

For the second term on the right-hand side of \eqref{eqn:lap-int-2}, we get
\begin{equation}\label{eqn:second}
   \left| \int_{\mathcal{C} \setminus B_\epsilon(x^*)} e^{-\frac{\theta(t,x)}t} f(x) dx \right|
  \leq \int_{\mathcal{C} \setminus B_\epsilon(x^*)} e^{-\frac{\theta(t,x^*) + \delta}t} |f(x)| dx
  \leq e^{-\frac{\delta}t} e^{-\frac{\theta(t,x^*)}t} \int_\mathcal{C} |f(x)| dx.
\end{equation}
As a result, the second term is exponentially small (at the rate $\delta$) as $t \to 0^+$ compared to the expansion \eqref{eqn:lap-int-1} obtained for the first term, hence it does not contribute to the asymptotic expansion. 

Finally, by \eqref{eqn:lap-int-2}, \eqref{eqn:first} and \eqref{eqn:second} we obtain the Laplace expansion \eqref{eqn:lap-int-1} by rewriting the expressions for the right-hand side of \eqref{eqn:first} in the $x$-coordinates.
\end{proof}

%
%

\end{document}